\documentclass[runningheads]{llncs}
\usepackage[T1]{fontenc}
\usepackage{amsmath}
\usepackage{multicol}
\usepackage{amsfonts}
\usepackage{xcolor}
\usepackage{graphicx}
\usepackage{tcolorbox}
\usepackage{float}
\usepackage{tikz}
\usepackage{tikz}
\usetikzlibrary{positioning}
\tikzset{edge/.style={draw, -stealth, thick}}
\usepackage{booktabs}
\usepackage{makecell}
\usepackage[gen]{eurosym}
\usepackage{tcolorbox}
\usepackage{tikz}
\usepackage[linesnumbered, ruled, vlined]{algorithm2e}
\usepackage{caption}
\usepackage{subcaption}
\usepackage{cite}
\usepackage{caption}
\usepackage{url}
\usepackage{hyperref}
\hypersetup{
    colorlinks=true,
    linkcolor=blue,
    filecolor=magenta,      
    urlcolor=cyan,
}
\begin{document}
\title{Group Trip Planning Query Problem with Multimodal Journey}
\author{Dildar Ali  \and Suman Banerjee \and Yamuna Prasad}
\authorrunning{Ali et al.} 
\institute{Indian Institute of Technology Jammu, J \& K-181221, India \email{\{2021rcs2009,suman.banerjee,yamuna.prasad\}@iitjammu.ac.in}}
\maketitle

\begin{abstract}
In \emph{Group Trip Planning (GTP) Query Problem}, we are given a city road network where a number of \emph{PoInt of Interest (PoI)} has been marked with their respective categories (e.g., Cafeteria, Park, Movie Theater, etc.). A group of agents want to visit one PoI from every category from their respective starting location and once finished, they want to reach their respective destinations. This problem asks which PoI from every category should be chosen so that the aggregated travel cost of the group is minimized. This problem has been studied extensively in the last decade, and several solution approaches have been proposed. However, to the best of our knowledge, none of the existing studies have considered the different modalities of the journey, which makes the problem more practical. To bridge this gap, we introduce and study the GTP Query Problem with Multimodal Journey in this paper. Along with the other inputs of the GTP Query Problem, we are also given the different modalities of the journey that are available and their respective cost. Now, the problem is not only to select the PoIs from respective categories but also to select the modality of the journey. For this problem, we have proposed an efficient solution approach, which has been analyzed to understand their time and space requirements. A large number of experiments have been conducted using real-life datasets, and the results have been reported. From the results, we observe that the PoIs and modality of journey recommended by the proposed solution approach lead to much less time and cost than the baseline methods. 
\end{abstract}

\section{Introduction}
In recent times, due to the advancement of wireless internet and GPS-enabled hand-holding mobile devices, capturing the location of a moving object has become easier. This leads to the generation of a large number of datasets and such datasets lead to a different domain called \emph{Spatial Databases and Spatial Data Mining} \cite{shekhar2007spatial,li2005trip}. One well-studied problem in the domain of spatial databases is the \emph{Group Trip Planning Query Problem}\cite{li2005trip,hashem2013group,ahmadi2015mixed,tabassum2017dynamic}. In this problem, we are given a road network of a city where the vertex set is constituted by the set of PoIs and the edge set is the road fragments connecting the PoIs. Also, the PoIs are classified into different categories, such as cafeterias, parks, movie theaters, restaurants, etc. A group of friends (referred to as agents in this paper) want to travel from their respective starting location to a destination location in the city, and during their journey, they want to visit one PoI from every category so that the aggregated travel cost by the group is minimized. This problem has been studied extensively in the last decade, and these can be broadly categorized into passenger-focused and transport manager-focused studies. The passenger-focused studies aim to optimize travel efficiency and experience by minimizing factors like travel time, cost, and number of transfers while considering comfort and convenience. Potthoff and Sauer \cite{potthoff2022efficient,sauer2020efficient} introduced the McTB approach, which optimizes three key criteria: arrival time, the number of public transit trips, and unrestricted transfer modes. Building on this, their HydRA algorithm \cite{potthoff2022efficient} enhances query execution efficiency, supporting faster computations in multimodal networks. Delling et al. \cite{delling2009engineering,delling2008timetable} developed efficient algorithms for multimodal routing and delays, incorporating transfer penalties and mode preferences. Further, RAPTOR, proposed by Delling et al. \cite{delling2015round}, focuses on minimizing transfer times and ensuring scalability in large public transportation networks. The transport manager-focused studies address system-wide efficiency, emphasizing operational optimization. Ceder \cite{ceder2016public} outlined methods for public transit scheduling and operations essential for integrating multiple transport modes.

\par Li et al. \cite{li2005trip} was the first to study the trip planning query problem on metric graphs and proposed several approximate solutions. Subsequently, this problem has been extended to include a group of travelers instead of a single traveler by Hashem et al. \cite{hashem2013group}. Later, Ahmadi et al. \cite{ahmadi2015mixed} proposed a mixed search (both breadth and depth) strategy and used the progressive group neighbor exploration technique. Lee and Park  \cite{lee2020collective} studied the trip planning problem to identify a typical meeting point such that the ride-sharing mechanism becomes effective. Additionally, several variants of the GTP Query Problem have been studied, such as incorporating PoI utility values \cite{barua2017weighted} and fairness criteria \cite{solanki2023fairness}. 
\par To the best of our knowledge, the GTP Query Problem has not been studied considering the presence of multiple transport mediums. This paper bridges this gap by studying the GTP Query Problem with multiple transport mediums. In particular, we make the following contributions in this paper:
\begin{itemize}
\item We introduce and study a practical variant of the GTP Query Problem considering the presence of multiple transport mediums. To the best of our knowledge, this is the first study in this direction.
\item To address this problem, we introduce an optimal journey planning algorithm and analyze it to understand its time and space requirements.  
\item A large number of experiments have been carried out to establish the effectiveness and efficiency of the proposed solution approach.
\end{itemize} 

\par The rest of this paper has been organized as follows. The background concepts and formal problem description have been described in Section \ref{Sec:BPD}. The proposed solution approach has been stated in Section \ref{Sec:PA}. Section \ref{Sec:EE} contains the experiment evaluation of the proposed solution approach. Finally, Section \ref{Sec:CFD} concludes our study and gives future research directions.

\section{Background and Problem Definition} \label{Sec:BPD}
\subsection{Road Network} 
In this study, we model a city road network using an undirected, weighted graph $\mathcal{G}(\mathcal{V}, \mathcal{E}, W)$, where the vertex set $\mathcal{V}=\{v_1, v_2, \ldots, v_n\}$ contains the set of PoIs in the city. The edge set $\mathcal{E}$ of $\mathcal{G}$ is constituted by the road segments joining the PoIs. Here, $W$ is the edge weight function that maps each edge to the corresponding travel cost between the two PoIs joined by the corresponding road segment, i.e., $W: \mathcal{E} \longrightarrow \mathbb{R}^{+}$. For any edge $(v_iv_j)$, its weight is denoted by $W(v_iv_j)$. As per our problem context, `vertex' and `PoI' have been used interchangeably. Similarly, `road segment ' and `edge' have been used interchangeably. For any vertex $v_i \in \mathcal{V}$, the neighbor of $v_i$ is denoted by $N(v_i)$ and defined as the other PoIs that are directly linked with $v_i$, i.e., $N(v_i)=\{v_j: (v_iv_j) \in \mathcal{E}\}$. The cardinality of $N(v_i)$ is called the degree of $v_i$. A sequence of vertices $P=<v_i,v_{i+1}, \ldots, v_j>$ is said to be a path in $\mathcal{G}$ if $(v_kv_{k+1}) \in \mathcal{E}$ for all $i \leq k \leq j-1$. For any path $P$ in $\mathcal{G}$, let $\mathcal{V}(P)$ and $\mathcal{E}(P)$ denote the set of vertices and edges that constitute the path. The weight of the path $P$ is defined as the sum of the edge weights of the edges that constitute the path, i.e., $W(P)=\underset{(v_pv_q) \in \mathcal{E}(P)}{\sum} \ W(v_pv_q)$. For any two vertices $v_i$ and $v_j$, the shortest path distance from $v_i$ to $v_j$ is denoted by $dist(v_iv_j)$. As $\mathcal{G}$ is undirected, $dist(v_iv_j)=dist(v_jv_i)$.
\subsection{Group Trip Planning Query Problem} \label{Sec:GTP}
For any positive integer $k$, $[k]$ denotes the set $\{1,2, \ldots, k\}$. In the GTP Query Problem, a set of $\ell$ many agents $\mathcal{U}=\{u_1, u_2, \ldots, u_{\ell}\}$ wants to travel from their respective source to the destination location. The source and destination location of the $i$-th agent is denoted by $v^{s}_i$ and $v^{d}_i$, respectively. All the PoIs under consideration can be classified into one of $k$ distinct categories. This can be formalized by the function $\mathcal{C}: \mathcal{V} \longrightarrow [k]$. For any $i \in [k]$, let $\mathcal{V}_{i}$ denotes the set of PoIs belongs to $i$-th category. Also, we assume that there does not exist any $i$ such that $\mathcal{V}_{i}=\emptyset$. For simplicity, we assume that the source and destination locations of the agents are also part of the city road network, and each one is a PoI in $\mathcal{G}$. Hence, $\mathcal{V}=\{v^{s}_1, v^{s}_2, \ldots, v^{s}_{\ell}\} \cup \{v^{d}_1, v^{d}_2, \ldots, v^{d}_{\ell}\} \cup \mathcal{V}_{1} \cup \mathcal{V}_{2} \cup \ldots \cup \mathcal{V}_{k}$. Now, in our problem context, we define \emph{Valid Path} in Definition \ref{Def:1}.
\begin{definition}[Valid Path] \label{Def:1}
For the agent $u_i$ a path $<v^{s}_i, v_1, v_2, \ldots, v_k, v^{d}_i>$ in $\mathcal{G}$ is said to be a valid path if $v^{s}_i$ and $v^{d}_i$ are the source and destination location of $u_i$, and for all $j \in [k]$, $v_j \in \mathcal{V}_{j}$.
\end{definition}

We observed that in a valid path, every agent visits one PoI from every category in the predetermined sequence and from the first to the $k$-th category, the path will be the same for all the agents. Consider for any $i \in [k]$, $|\mathcal{V}_{i}|=n_i$. Now, the number of valid paths will be $\underset{i \in [k]}{\prod} n_i$. Let $\mathcal{P}$ contain the common portions set of all valid paths. Now, for any valid path $p \in \mathcal{P}$, let $\mathcal{D}(p)$ denote the aggregated distance for the path. If the common path is $p=<v_1, v_2, \ldots, v_k>$ then the aggregated distance $\mathcal{D}(p)$ can be computed using Equation \ref{Eq:1}. 

\begin{equation} \label{Eq:1}
\mathcal{D}(p)= \underset{i \in [\ell]}{\sum} dist(v^{s}_i v_1) + \underset{i \in [k-1]}{\sum} dist(v_i v_{i+1}) + \underset{i \in [\ell]}{\sum} dist(v_k v^{d}_i)
\end{equation} 


Next, we formally define the GTP Query Problem, which has been formally stated in Definition \ref{Def:2}.
\begin{definition}[GTP Query Problem] \label{Def:2}
Given a city road network $\mathcal{G}(\mathcal{V}, \mathcal{E}, W)$ which contains all the PoIs of interest and classified into different categories, the GTP Query Problem asks to select one PoI from each category such that aggregated travel cost by the group is minimized. This can be expressed using Equation \ref{Eq:2}.

\begin{equation} \label{Eq:2}
p^{*} \longleftarrow \underset{p \in \mathcal{P}}{argmin} \ \mathcal{C}(p)
\end{equation}

Here, $p^{*}=<v^{*}_1, v^{*}_2, \ldots, v^{*}_k>$ denotes the optimal path.
\end{definition}

\subsection{Multi Modal Journey}
Consider in the city of our consideration, there exists $p$ different modes of transport. In this paper, we use the terminology `\emph{vehicle}' to refer to any conveyance used in any mode of transport. For any $i \in [p]$, let $V_{i}$ denote the set of vehicles of $i$-th transport medium. Next, we state the notion of \emph{route} in Definition \ref{Def:Route}. 

\begin{definition}[Route] \label{Def:Route}
Given a city road network $\mathcal{G}(\mathcal{V}, \mathcal{E}, W)$ the route is defied as a path in the network such that there must exists at least one vehicle that covers that path in its journey.
\end{definition}
For any route $r=<v_i, v_{i+1}, \ldots, v_j>$, let $PoI(r)$ denotes the set of PoIs that the route $r$ covers. $PoI_{s}(r)$ and $PoI_{e}(r)$ denotes the starting and ending PoI of the route $r$. For any vehicle $x$ that covers the route $r$, $t^{x}_{s}(r)$ and $t^{s}_{e}(r)$ denotes the start time and end time of the vehicle $x$ for the route $r$, respectively. 

\begin{definition}[Journey Planning] \label{Def:JP}
Given the city road network $\mathcal{G}(\mathcal{V}, \mathcal{E}, W)$ with the transport details and cost, two PoIs of $\mathcal{G}$ $v_a$ and $v_b$, and a start time $t$ the journey planning is defined as a sequence of routes along with the corresponding vehicles such that the following criteria got satisfied:
\begin{itemize}
\item The starting PoI of the first route and the ending PoI of the last route in the journey plan must be the start and destination of the journey.
\item For any two consecutive routes, the starting time of the chosen vehicle of the second route should be more than the ending time of the chosen vehicle of the first route.  
\end{itemize}
\end{definition}

In this study, we consider that for a specific modality (say $i$) for any two vehicles $x,y \in V_{i}$, such that both of them cover the route $r$, the cost of commuting between the PoIs $v_a$ and $v_b$ by either $x$ or $y$ will be the same. Now, we define the cost of a journey planning in Definition \ref{Def:Cost}.

\begin{definition}[Cost of a Journey Planning] \label{Def:Cost}
Given the transport details, the cost, and journey planning (i.e., the routes and the corresponding vehicles as stated in Definition \ref{Def:JP}), the cost of journey planning is defined as the sum of the costs of the routes that constitute the whole journey. For the journey planning $\mathcal{R}=<r_1,r_2, \ldots, r_y>$, and $V=\{i,j, \ldots, k\}$ the cost is denoted by $\mathcal{C}(\mathcal{R})$ and can be mathematically expressed as:
\begin{equation}
\mathcal{C}(\mathcal{R})= \underset{r \in \mathcal{R}}{\sum} \mathbb{C}(a,b,c)
\end{equation}
Where $v_a$ and $v_b$ are the start and end PoI of the route $r$, and in the journey planning, this route will be traveled using a vehicle of $c$-th modality.
\end{definition}
%

\begin{definition}[Minimum Cost Journey Planning] \label{Def:Min_Cost}
Given a city road network $\mathcal{G}(\mathcal{V}, \mathcal{E}, W)$, the transport details (i.e., different modalities of journey with their respective routes and their associated cost) the minimum cost journey planning problem asks to choose a least cost journey planning. Mathematically, this problem can be posed as an optimization problem as follows:   
\begin{equation}
\mathcal{R}^{*} \longleftarrow \underset{\mathcal{R} \in \mathbb{R}}{argmin} \ \mathcal{C}(\mathcal{R})
\end{equation}
Here, $\mathbb{R}$ denotes the set of all possible journey planning. 
\end{definition}
Subsequently, we define the GTP Query Problem under the realm of multi-modal journey and formally define our problem.

\subsection{GTP Query Problem with Multi Modal Journey}
In this paper, we introduce a variant of the GTP Query Problem where multiple modes of transportation exist in the city under consideration. Now, it can be observed for every agent, the individual cost may be different, and this is due to the following reasons. Let, $<v^{*}_1, v^{*}_2, \ldots, v^{*}_k>$ the common path, however, for any $i$-th agent the path that it covers is $<v^{s}_{i}, v^{*}_1, v^{*}_2, \ldots, v^{*}_k, v^{d}_{i}>$. So, to commute from their respective source locations to $v^{*}_1$ and from $v^{*}_k$ to their respective destination location incur different costs. As the common portion of the trip, the modality of the journey has to be decided collectively. Now, we define the cost of commuting for the whole group in Definition \ref{Def:Group_Cost}.

\begin{definition}[The Cost of the Group]\label{Def:Group_Cost}
Given a GTP Query Problem instance, any of the transport details of the city along with the cost, the cost of commuting of the group can be defined as the sum of the following three costs:
\begin{itemize}
\item Sum of the individual costs for reaching the recommended PoI of the first category.
\item Sum of the total cost incurred by the group from the first category of PoI to the $k$-th category of PoI.
\item  Sum of the individual costs for reaching the recommended destinations PoI from the $k$-th category of PoI.
\end{itemize}
Hence, mathematically this can be posed as follows 
\begin{equation} \label{Eq:Group_Cost}
\scriptsize
\mathcal{C}_{R}(\mathcal{U})= \underset{i \in [\ell]}{\sum} \mathbb{C}(v^{s}_i, v^{*}_1,d_i) + \ell \cdot \underset{j \in [k-1]}{\sum} \mathbb{C}(v^{*}_j, v^{*}_{j+1},d_j) + \underset{s \in [\ell]}{\sum} \mathbb{C}(v^{*}_k,v^{d}_i,d_s)
\end{equation}

\end{definition}
It can be observed that in Equation \ref{Eq:Group_Cost}, a subscript $R$ has been used. This signifies this cost has been defined for a specific journey planning $R$ (for all the agents collectively). Now, based on the group cost as defined in Equation \ref{Eq:Group_Cost}, we formally define the GTP Query with Multi-Modal Journey Problem in Definition \ref{Def:Problem}.

\begin{definition}[GTP Query with Multi Modal Journey Problem] \label{Def:Problem}
Given a GTP Query Problem instance along with the transport details of the city (i.e., details of different transport medium, routes, their corresponding costs, etc.), the GTP Query with Multi-Modal Journey Problem asks to recommend one PoI from every category, and the journey plan for all the agents such that the total travel cost as defined in Equation \ref{Eq:Group_Cost} gets minimized. Mathematically, this problem can be posed as a discrete optimization problem as follows.
\begin{equation}
\mathcal{R}^{*} \longleftarrow \underset{R \in \mathbb{R}}{argmin} \ \mathcal{C}_{R}(\mathcal{U})
\end{equation}

\end{definition}

\section{Proposed Approach} \label{Sec:PA}
Given the city road network $\mathcal{G}(\mathcal{V}, \mathcal{E}, W)$, the transport details, and the set of PoIs, we propose a dynamic programming (DP) based solution approach that will return the minimum cost journey plan. The working of Algorithm \ref{Alg:optimal_journey} is as follows. In Line No. $1$ to $5$, we construct the multi-graph using the PoIs and initialize the categories of PoIs, source, destinations, and the DP dictionary. Next, in Line No. $6$ to $19$, the transition between source PoIs to the PoIs of the first intermediate category is performed. The transition between intermediate categories is performed, and a common path is selected in Line No. $20$ to $33$. Next, in Line No. $34$ to $45$, the transitions between the last intermediate category and the destinations are performed. Finally, in Line No. $46$ to $47$, the total minimum cost and the shortest paths \cite{dijkstra2022note} of each individual agent are extracted. In the intermediate paths, all the agents travel together, and the intermediate travel cost is multiplied by the number of agents $(\ell)$. The representation of PoIs in Algorithm \ref{Alg:optimal_journey} is as follows. Consider a set of $\ell$ many agents $\{v_{1}^{s}, v_{2}^{s}, \ldots v_{\ell}^{s}\}$ are the source vertices and $\{v_{1}^{d}, v_{2}^{d}, \ldots v_{\ell}^{d}\}$ are the destination vertices in the network $\mathcal{G}$ and the PoIs are categories into $k$ types and they contains $p_{1}, p_{2}, \ldots,p_{k}$ many PoIs, respectively. Now, from source vertices to the first category of PoIs, each PoI contains $g \times \ell$ many variables to store the transport medium and costs where $g$ is the number of transport mediums. So, the first category contains total $p_{1} \times g \times \ell$ many variables. It can be represented as a matrix below.

\begin{minipage}{0.4\textwidth}
\[
\mathcal{V}_{11} =
\begin{bmatrix}
\mathcal{V}_{11}^{(1,\mathcal{V}_{1}^{s})} & \mathcal{V}_{11}^{(2,\mathcal{V}_{1}^{s})} & \ldots & \mathcal{V}_{11}^{(n,\mathcal{V}_{1}^{s})} \\
\mathcal{V}_{11}^{(1,\mathcal{V}_{2}^{s})} & \mathcal{V}_{11}^{(2,\mathcal{V}_{2}^{s})} & \ldots & \mathcal{V}_{11}^{(n,\mathcal{V}_{2}^{s})} \\
\vdots & \vdots & \vdots & \vdots  \\
\mathcal{V}_{11}^{(1,\mathcal{V}_{\ell}^{s})} & \mathcal{V}_{11}^{(2,\mathcal{V}_{\ell}^{s})} & \ldots & \mathcal{V}_{11}^{(n,\mathcal{V}_{\ell}^{s})} \\
\end{bmatrix}
\]
\end{minipage}
\hfill
\begin{minipage}{0.5\textwidth}
\[
\mathcal{V}_{21} =
\begin{bmatrix}
\mathcal{V}_{21}^{(1,\mathcal{V}_{11})} & \mathcal{V}_{21}^{(2,\mathcal{V}_{11})} & \ldots & \mathcal{V}_{21}^{(n,\mathcal{V}_{11})} \\
\mathcal{V}_{21}^{(1,\mathcal{V}_{12})} & \mathcal{V}_{21}^{(2,\mathcal{V}_{12})} & \ldots & \mathcal{V}_{21}^{(n,\mathcal{V}_{12})} \\
\vdots & \vdots & \vdots & \vdots  \\
\mathcal{V}_{21}^{(1,\mathcal{V}_{1p_{1}})} & \mathcal{V}_{21}^{(2,\mathcal{V}_{1p_{1}})} & \ldots & \mathcal{V}_{21}^{(n,\mathcal{V}_{1p_{1}})} \\
\end{bmatrix}
\]
\end{minipage}

 Here, $\mathcal{V}_{11}^{(1,\mathcal{V}_{1}^{s})}$ contains the transport cost from vertex $\mathcal{V}_{1}^{s}$ to vertex $\mathcal{V}_{11}$ in the first category of PoIs via transport medium $1$. Next, from the first category to the second category of PoIs, each PoI in the second category obtains $g$ many variables. The second category contains total $p_{2} \times g$ variables, and the cost of each transport medium will be multiplied with $\ell$ as a group travels all the agents. Here, $\mathcal{V}_{21}$ is a PoI in the second category. Similarly, for all intermediate categories till the $k^{th}$ category, all PoIs contain $g$ many variables. Finally, from the last category to the destination vertices, each obtains $g$ many variables. In this case, the destination vertices cost will not be multiplied by the number of agents as their destinations differ. 

\SetKwComment{Comment}{/* }{ */}
\begin{algorithm}
\scriptsize
\caption{An Optimal Journey Planning Algorithm}
\label{Alg:optimal_journey}
\KwData{The Road Network $\mathcal{G}(\mathcal{V}, \mathcal{E}, W)$, $\ell ~\text{number of agents}$, transport details, the set of Source and Destination PoInts of Interest (PoIs) $v^{s}_i$ and $v^{d}_i$, the set of PoI categories $\mathcal{C}$, and the set of transport mediums $M$.}
\KwResult{A journey (routes and corresponding vehicles) connecting the PoIs $v^{s}_i$ and $v^{d}_i$ with minimum cost.}

\Comment{Step 1: Graph Construction and Initialization}
Construct a multi-graph $\mathcal{G}$ using transport details, with nodes $\mathcal{V}$ and edges $\mathcal{E}$\;
Identify the set of categories $\mathcal{C} = \{C_1, C_2, \dots, C_k\}$ from node attributes\;
Categorize PoIs into $\mathcal{P}[c]$ for each category $c \in \mathcal{C}$, ensuring intermediate categories have a fixed number of PoIs\;
Randomly assign sources $v^{s}_i \in C_1$ and destinations $v^{d}_i \in C_k$\;

Initialize a DP dictionary $DP[c][i] \leftarrow \infty$ for minimum costs to PoIs in category $c$ from source $v^{s}_i$\;

\Comment{Step 2: Base Case (First Category)}
$First\_Category\_Cost = \infty$, $Chosen\_PoI = \emptyset$, $Best\_Path = [~]$\;
\For{each PoI $j \in \mathcal{P}[C_1]$}{
$total\_cost = 0$;
$path = [~]$\;
\For{each source $i$ in $v_i^s$}{
\If{a path exists from $i$ to $j$}{
$path \leftarrow \text{ShortestPath}(\mathcal{G}, i, j, W)$\;
$path\_cost \leftarrow \text{ShortestPathCost}(\mathcal{G}, i, j, W)$\;
$total\_cost \leftarrow total\_cost + path\_cost$\;
$update ~path$\;
}}
\If{$total\_cost < First\_Category\_Cost$}{
$First\_Category\_Cost \leftarrow total\_cost$\;
$Chosen\_PoI \leftarrow  j$\;
$Best\_Path \leftarrow path$
}}
$DP[C1][j] \longleftarrow \{
        PoI \leftarrow Chosen\_PoI,
        cost \leftarrow First\_Category\_Cost,
        paths \leftarrow Best\_Path
    \}$

\Comment{Step 3: Transition for Intermediate Categories}
\For{each category $c \in \{2, \dots, k-1\}$}{
$min\_cost = \infty$, $Chosen\_PoI = \emptyset$, $Best\_Path = [~]$\;
\For{each PoI $j \in \mathcal{P}[C_c]$}{
\For{each PoI $i \in \mathcal{P}[C_{c-1}]$}{
\If{a path exists from $i$ to $j$}{
$path \leftarrow \text{ShortestPath}(\mathcal{G}, i, j, W)$\;
$path\_cost \leftarrow DP[c-1][j][cost] + \text{ShortestPathCost}(\mathcal{G}, i, j, W)$\;
$total\_cost \leftarrow total\_cost + path\_cost$\;
$update ~path$\;
}}
\If{$total\_cost < min\_cost$}{
$min\_cost \leftarrow total\_cost$\;
$Chosen\_PoI \leftarrow  j$\;
$Best\_Path \leftarrow path$
}}
$DP[c][j] \longleftarrow \{
        PoI \leftarrow Chosen\_PoI,
        cost \leftarrow min\_cost,
        paths \leftarrow Best\_Path
    \}$
}

\Comment{Step 4: Transition to Destinations}
$last\_PoI = DP[k-1][j][PoI]$\;
$optimal\_paths = [~]$\;
$group\_path \leftarrow DP[k-1][j][\text{path}]$\;
$shared\_cost \leftarrow DP[k-1][j][\text{cost}]$\;

$group\_trip\_cost \leftarrow shared\_cost \times \ell$\;

$individual\_cost\_sum \leftarrow 0$\;

\For{each destination $ j \in v_i^d$}{
\If{a path exists from $last\_PoI$ to $j$}{
$path \leftarrow group\_path + \text{ShortestPath}(\mathcal{G}, last\_PoI, j, W)$\;
$individual\_cost \leftarrow \text{ShortestPathCost}(\mathcal{G}, last\_PoI, j, W)$\;
$individual\_cost\_sum \leftarrow individual\_cost\_sum  +individual\_cost$\;
$optimal\_paths.append(path)$\;         
}}

\Comment{Step 5: Final Total Cost for Group}
$total\_cost\_for\_all\_agents \leftarrow group\_trip\_cost + individual\_cost\_sum$\;

\Return $total\_cost\_for\_all\_agents, optimal\_paths$\;
\end{algorithm}
\paragraph{\textbf{Complexity Analysis of Algorithm \ref{Alg:optimal_journey}}} 
\par Now, we analyze the time and space requirements of Algorithm \ref{Alg:optimal_journey}. In-Line No. $1$, constructing the multi-graph and ensuring connectivity by adding edges between disconnected components takes $\mathcal{O}(n+m)$. Next, in Line No. $2$ to $3$ scanning node attributes to classify PoIs into $k$ categories takes $\mathcal{O}(n)$ and in Line No. $4$ randomly assigning source $v^{s}_i$ and destination $v^{d}_i$ will take $\mathcal{O}(1)$ time. Initializing the DP dictionary in Line No. $5$ will take $\mathcal{O}(n)$, and in Line No. $6$, initialization of variables will take $\mathcal{O}(1)$ time. So, Line No. $1$ to $6$ will take $\mathcal{O}(n + m)$ time. Next, in Line No. $7$ to $19$ for each PoI $j \in \mathcal{P}[C_1]$ and each source $i \in v^{s}_i$, computing shortest path (using Dijkstra) will take $\mathcal{O}(n \log n +  m \log n)$ and the total time requirement will be $\mathcal{O}(\ell \cdot p_{1} \cdot n \log n + \ell \cdot p_{1} \cdot m \log n)$, where $\ell$ and $p_{1}$ are the number of agents and number of PoIs in $\mathcal{P}[C1]$, respectively. In Line No. $20$ to $33$, for each PoI $j \in \mathcal{P}[C_c]$ and each PoI $i \in \mathcal{P}[C_{c-1}]$ computing shortest path will take $\mathcal{O}(n \log n + m \log n)$. Considering $p_{c}$ is the number of PoIs in $\mathcal{P}[C_{c}]$ and $p_{c-1}$ is the number of PoIs in $\mathcal{P}[C_{c-1}]$, the time requirements will be $\mathcal{O}(p_{c} \cdot p_{c-1} \cdot (n \log n + m \log n))$. Now, summing all the intermediate category $(k-2)$, the total time requirements will be $\mathcal{O}((k-2) \cdot p \cdot p \cdot (n+m)\log n)$ i.e., $\mathcal{O}(k \cdot p^{2} \cdot (n+m)\log n)$. Next, in Line No. $34$ to $47$ time requirements will be $\mathcal{O}(\ell \cdot p_{k} \cdot (n+m) \log n)$, where $p_{k}$ is the number of PoIs in the last category. So, the overall time requirement for Algorithm \ref{Alg:optimal_journey} will be $\mathcal{O}(k \cdot p^{2} \cdot (n+m)\log n)$.
\par Now, the graph $\mathcal{G}$ uses $\mathcal{O}(n + m)$ space, and the DP table stores costs for all PoIs in all categories will be $\mathcal{O}(k \cdot p)$. Next, each path required $\mathcal{O}(p \cdot s)$ space, where $s$ is the average length of the paths. Hence, the overall space requirements of Algorithm \ref{Alg:optimal_journey} will be $\mathcal{O}(n + m + k \cdot p + p \cdot s)$ i.e., $\mathcal{O}(n + m + k \cdot p)$.

\paragraph{\textbf{An Illustrative Example}}
 We consider the following scenario to demonstrate the working of the Algorithm \ref{Alg:optimal_journey}. Let, we have a road network $\mathcal{G}(\mathcal{V}, \mathcal{E}, W)$ consists of 10 Point of Interest (PoIs) $v_1, v_2, \dots, v_{10}$. The PoIs are categorized into five categories $\mathcal{C} = \{C_1, C_2, C_3, C_4, C_5\}$, where $C_1 = \{v_1, v_2\}$ (sources), $C_2 = \{v_3, v_4\}$ (first category), $C_3 = \{v_5, v_6\}$ (second category), $C_4 = \{v_7, v_8\}$ (third category), $C_5 = \{v_9, v_{10}\}$ (destinations). The edge weights $W$ represent the travel costs, and each edge is labeled with multiple transport mediums like Bus (B), Car (C), Train (T), and Ferry (F). Figure \ref{fig:example_graph} shows the road network and categories. Next, we initialize $DP[C_1][v_1] = 0$, $DP[C_1][v_2] = 0$ and set all other $DP[c][v_i] = \infty$. In the transition for category $C_2$, for each PoI in $C_2$, compute the minimum Cost from PoIs in $C_1$ considering all transport mediums. For example, for $v_3$, the Cost from $v_1$ via Bus is $5$, and via Train is $7$. The Cost from $v_2$ via Car is $8$, and via Ferry is $12$. Next, for $v_4$, cost from $v_1$ via Car is $3$, via Bus is $6$. The Cost from $v_2$ via Ferry is $10$, and via Train is $9$. So, minimum cost if $v_{1}$ and $v_{2}$ want to select $v_{3}$ is $(8+5) = 13$ and if $v_{4}$ is chosen then cost will be $(10+3)= 13$. The updated value in DP table will be $DP[C_2][v_3] = 13$ and $DP[C_2][v_4] = 13$. Now, from the second category to the third category, the minimum Cost for $v_{3}$ to $v_{5}$ is $4$ via Train and $v_{3}$ to $v_{6}$ via Bus is $6$. Similarly, the cost for $v_{4}$ to $v_{5}$ and $v_{6}$ are $7$ and $5$, respectively. So, updated value is $DP[C_3][v_5] = 17$ and $DP[C_3][v_6] = 18$. For the fourth category the updated values are $DP[C_4][v_7] = 20$  and $DP[C_4][v_8] = 21$. Finally in the last category $DP[C_5][v_9] = 25$ and $DP[C_5][v_{10}] = 24$. So, in the intermediate category, two possible minimum cost paths are $v_3 \to v_5 \to v_7$ and $v_4 \to v_6 \to v_7$. Both the paths take Cost = $20$, and these intermediate paths are common for all the agents. Hence, for agent $1$ who starts journey from $v_{1}$ to $v_{10}$ the possible path is $v_{1} \to v_3 \to v_5 \to v_7 \to v_{10}$ and the cost is $29$. Similarly for agent $2$, possible path will be  $v_{2} \to v_3 \to v_5 \to v_7 \to v_{9}$ and travel cost is $31$. So, the total group trip cost will be $(20+(5+4+6+5)) = 40$. 

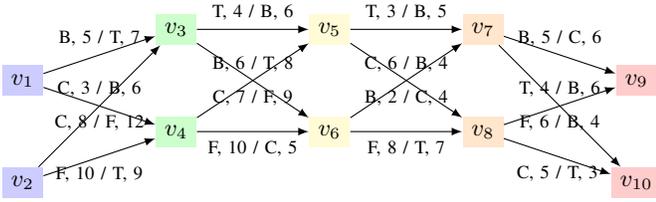
\begin{figure}[h]
\centering
\begin{tikzpicture}[scale=0.68, every node/.style={font=\small}]
    \node[fill=blue!20] (v1) at (0, 4) {$v_1$};
    \node[fill=blue!20] (v2) at (0, 2) {$v_2$};
    \node[fill=green!20] (v3) at (3, 5) {$v_3$};
    \node[fill=green!20] (v4) at (3, 3) {$v_4$};
    \node[fill=yellow!20] (v5) at (6, 5) {$v_5$};
    \node[fill=yellow!20] (v6) at (6, 3) {$v_6$};
    \node[fill=orange!20] (v7) at (9, 5) {$v_7$};
    \node[fill=orange!20] (v8) at (9, 3) {$v_8$};
    \node[fill=red!20] (v9) at (12, 4) {$v_9$};
    \node[fill=red!20] (v10) at (12, 2) {$v_{10}$};

    \draw[-latex] (v1) -- (v3) node[midway, above, font=\scriptsize] {B, 5 / T, 7};
    \draw[-latex] (v1) -- (v4) node[midway, below, font=\scriptsize] {C, 8 / F, 12};
    \draw[-latex] (v2) -- (v3) node[midway, above, font=\scriptsize] {C, 3 / B, 6};
    \draw[-latex] (v2) -- (v4) node[midway, below, font=\scriptsize] {F, 10 / T, 9};
    \draw[-latex] (v3) -- (v5) node[midway, above, font=\scriptsize] {T, 4 / B, 6};
    \draw[-latex] (v3) -- (v6) node[midway, below, font=\scriptsize] {C, 7 / F, 9};
    \draw[-latex] (v4) -- (v5) node[midway, above, font=\scriptsize] {B, 6 / T, 8};
    \draw[-latex] (v4) -- (v6) node[midway, below, font=\scriptsize] {F, 10 / C, 5};
    \draw[-latex] (v5) -- (v7) node[midway, above, font=\scriptsize] {T, 3 / B, 5};
    \draw[-latex] (v5) -- (v8) node[midway, below, font=\scriptsize] {B, 2 / C, 4};
    \draw[-latex] (v6) -- (v7) node[midway, above, font=\scriptsize] {C, 6 / B, 4};
    \draw[-latex] (v6) -- (v8) node[midway, below, font=\scriptsize] {F, 8 / T, 7};
    \draw[-latex] (v7) -- (v9) node[midway, above, font=\scriptsize] {B, 5 / C, 6};
    \draw[-latex] (v7) -- (v10) node[midway, above, font=\scriptsize] {T, 4 / B, 6};
    \draw[-latex] (v8) -- (v10) node[midway, below, font=\scriptsize] {C, 5 / T, 3};
    \draw[-latex] (v8) -- (v9) node[midway, below, font=\scriptsize] {F, 6 / B, 4};
\end{tikzpicture}
\caption{Road network with categories $C_1$ to $C_5$.}
\label{fig:example_graph}
\end{figure}

\begin{theorem}
The Algorithm \ref{Alg:optimal_journey} always terminates within a finite number of steps for any valid input graph $\mathcal{G}$. 
\end{theorem}
\begin{proof}
Let $|\mathcal{V}|$ be the number of PoIs and $|\mathcal{E}|$ be the number of edges in the graph $\mathcal{G}$. Each PoI and edge from the input is processed once to construct $\mathcal{G}$, which takes $\mathcal{O}(|\mathcal{V}| + |\mathcal{E})$. Let $\mathcal{C}_1, \mathcal{C}_2, \dots, \mathcal{C}_p$ denote the $p$ connected components of the graph. At most $p-1$ edges are added to connect these components, which is a finite operation. For each category $\mathcal{C} = \{C_1, C_2, \dots, C_k\}$, the algorithm iterates through all pairs of PoIs $i \in C_{c-1}, j \in C_c$. Let $|C_c|$ denote the number of PoIs in category $C_c$. Then the pairwise iterations across all categories take at most $\sum_{c=2}^k |C_{c-1}| \cdot |C_c| \leq |\mathcal{V}|^2$. Next, the shortest path computations are performed using Dijkstra’s algorithm \cite{dijkstra2022note} with complexity \(O(|\mathcal{E}| \cdot \log|\mathcal{V}|)\) for each source-target pair and there are \(|C_{c-1}| \cdot |C_c|\) such pairs, the total complexity for path computations is $ \mathcal{O}(|\mathcal{V}|^2 \cdot |\mathcal{E}| \cdot \log|\mathcal{V}|)$. The dynamic programming updates for costs and paths iterate over \(|C_c|\) PoIs per category, taking \(O(|\mathcal{V}| \cdot |\mathcal{C}|)\). Each algorithm stage involves a finite number of operations. Thus, the Algorithm \ref{Alg:optimal_journey} terminates.
\end{proof}

\begin{theorem}
If the input graph $\mathcal{G}$ is connected the Algorithm \ref{Alg:optimal_journey} guarantees a feasible journey from sources $v^{s}_i$ to destinations $v^{d}_i$. 
\end{theorem}

\begin{proof}
Let $\mathcal{G} = (\mathcal{V}, \mathcal{E}, W)$ be the input graph. If $\mathcal{G}$ is not connected, the algorithm identifies connected components $\mathcal{C}_1, \mathcal{C}_2, \dots, \mathcal{C}_p$, which ensures $\bigcup_{i=1}^p \mathcal{C}_i = \mathcal{V}, ~ \mathcal{C}_i \cap \mathcal{C}_j = \emptyset \text{ for } i \neq j $. For each pair of disconnected components, an edge is added between a PoI $v_i \in \mathcal{C}_i\) and $v_j \in \mathcal{C}_j$. At most $p-1$ edges are added, ensuring that the graph becomes connected, i.e., $\text{diameter}(\mathcal{G}) < \infty$. For each source $v^{s}_i$, the algorithm computes the shortest path to every PoI in $C_1$ using Dijkstra's algorithm, which guarantees $\exists \text{path}(v^{s}_i, v_j)~ \forall v_j \in C_1 $. By induction on the categories, we assume a path exists from $v^{s}_i$ to any $v_j \in C_{c-1}$. The algorithm computes the shortest path from $v_j$ to all $v_k \in C_c$, ensuring $\exists \text{path}(v^{s}_i, v_k)~ \forall v_k \in C_c$. In the final PoI to destinations, the algorithm computes the shortest path from the last PoI in $C_k$ to each destination $v^{d}_i$, ensuring $\exists \text{path}(v^{s}_i, v^{d}_i)~ \forall v^{d}_i \in \mathcal{V}$. Thus, a feasible journey exists for all source-destination pairs.
\end{proof}

\begin{theorem}
The Algorithm \ref{Alg:optimal_journey} computes a journey with the minimum total cost between sources $v^{s}_i$ and destinations $v^{d}_i$. 
\end{theorem}

\begin{proof}
Let $\mathcal{G} = (\mathcal{V}, \mathcal{E}, W)$ be the input graph and $DP[c][j]$ represent the minimum cost to reach PoI $j \in C_c$ from any source $v^{s}_i$. The recurrence relation for $DP$ is $DP[c][j] = \min_{i \in C_{c-1}} \big(DP[c-1][i] + \text{ShortestPathCost}(i, j)\big)$. For the base case $(c=1)$, $DP[1][j] = \sum_{i \in v^{s}_i} \text{ShortestPathCost}(i, j)$. Now, applying mathematical induction on categories, we assume $DP[c-1][i]$ stores the minimum cost to reach $i \in C_{c-1}$. The Algorithm \ref{Alg:optimal_journey} computes the cost for $j \in C_c$ by evaluating all transitions $i \to j$ and taking the minimum, ensuring optimality $DP[c][j] = \min_{i \in C_{c-1}} (DP[c-1][i] + W(i, j))$. By mathematical induction, $DP[c][j]$ stores the minimum cost for all paths ending at $j \in C_c$. Now, for final transitions to the destinations, for each destination $v^{d}_i$, the total cost is$    \text{TotalCost} = \sum_{j \in C_k} DP[k][j] + \sum_{j \in v^{d}_i} W(j, v^{d}_i)$. Since all pairwise costs are minimized in the DP updates, the final cost is also minimized. Thus, the algorithm computes the globally optimal journey.
\end{proof}

\section{Experimental Evaluation} \label{Sec:EE}
\paragraph{\textbf{Dataset Description.}}
We evaluate our proposed approach with two different real-world datasets. First, the networks of Switzerland, which are previously used by Sauer et al. \cite{potthoff2022efficient, sauer2020efficient}, and the second is the road network of Helsinki\footnote{\url{https://welcome.hel.fi/}} city of Finland. The public transit network for Switzerland and Finland are collected from GTFS feed\footnote{\url{https://gtfs.geops.ch/}}. The Switzerland network covers the timetable of two business days, and the road networks were obtained from OpenStreetMap\footnote{\url{https://download.geofabrik.de/}}. The Finland dataset includes travel time and distance between all SYKE (Finnish Environment Institute), calculated for walking, cycling, public transport, and car travel \cite{tenkanen2020longitudinal}. In Switzerland network we merged seven different feeds `Bus', `Train', `Ferry', `Funicular', `Gondola', `Subway' and `Tram' into a single transport network. We categorize the PoIs into ten distinct categories `Train Station', `Public Square', `City Center', `Bridge', `School', `Park', `Bus Stop', `Airport', `Healthcare Facility' and `Hotel'. Further, we compute the travel cost for each medium of transport using (Travel Cost = Base Fare + (Cost per minute * Travel Time) + (Cost per meter * Travel Distance)) the information provided in Table \ref{tab:transport_fares}. In the case of the Finland dataset, there are complete transport road networks for the city of Helsinki. However, in our problem context, we have taken a small portion of the road networks that contain $1100$ unique PoIs. One point to be highlighted is the fares given in Table \ref{tab:transport_fares}, \ref{tab:transport_costs_helsinki} are the approximate fares we assumed for our problem context. An overview of the networks is given in the Table \ref{Table:Dataset}.
\vspace{-0.3in}
\begin{table}[h!]
\scriptsize
\caption{Description of Switzerland Dataset}
\centering
\tiny
\begin{tabular}{|c|c|c|c|}
\hline
\textbf{Transport Medium} & \makecell{\textbf{Base Fare} \\ \textbf{(CHF)}} & \makecell{\textbf{Cost per Meter} \\ \textbf{(CHF)}} & \makecell{\textbf{Cost per Minute} \\ \textbf{(CHF)}} \\ \hline
Bus & 2.50 - 4.00 & 0.01 - 0.03 & 0.05 - 0.10 \\ \hline
Tram & 2.50 - 4.00 & 0.01 - 0.03 & 0.05 - 0.10 \\ \hline
Train & $\sim$5.00 & 0.03 - 0.05 & 0.10 - 0.15 \\ \hline
Ferry & 5.00 - 10.00 & 0.05 - 0.10 & 0.15 - 0.25 \\ \hline
Funicular & 1.30 - 5.00 & 0.02 - 0.04 & 0.10 - 0.15 \\ \hline
Gondola & 5.00 - 15.00 & 0.05 - 0.15 & 0.20 - 0.50 \\ \hline
Subway & 2.50 - 4.00 & 0.01 - 0.03 & 0.05 - 0.10 \\ \hline
\end{tabular}
\label{tab:transport_fares}
\end{table}
\vspace{-0.4in}

\begin{table}[h!]
\centering
\begin{minipage}[t]{0.55\textwidth}
\tiny
\centering
\caption{Description of Finland Dataset}
\label{tab:transport_costs_helsinki}
\begin{tabular}{|l|c|c|c|}
\hline
\textbf{Transport Medium} & \makecell{\textbf{Base Fare} \\ \textbf{(EUR)}} & \makecell{\textbf{Cost per Meter} \\ \textbf{(EUR)}} & \makecell{\textbf{Cost per Minute} \\ \textbf{(EUR)}} \\ \hline
Bike & 0.00 - 5.00 & 0.00 - 0.10 & 0.05 - 0.10 \\ \hline
Public Transport & 3.20 & 0.03 - 0.05 & 0.05 - 0.10 \\ \hline
Private Car (Taxi) & 5.90 & 0.01 - 0.05 & 0.74 \\ \hline
\end{tabular}
\end{minipage}\hfill
\begin{minipage}[t]{0.4\textwidth}
\tiny
\centering
\caption{Dataset Description}
\label{Table:Dataset}
\begin{tabular}{ccc}
\hline
\textbf{Attributes} & \textbf{Switzerland} & \textbf{Finland}  \\ \hline
Stops  & 44557  &    1100 \\ 
Routes & 168294  & 4840000\\
Vertices & 1310  & 1100 \\
Edges & 11370  &  604950\\
\hline
\end{tabular}
\end{minipage}
\end{table}

\paragraph{\textbf{Environment and Key Parameter Setup.}} \label{Sec:Env_Setup} The proposed and baseline methods are implemented in Python using the Jupyter Notebook Platform. All the experiments are conducted in a Ubuntu-operated desktop system with 64 GB RAM and an Xeon(R) 3.5 GHz processor. Next, we vary the number of agents $(|\mathcal{U}|)$ by $5,10,20,50,100$ to show the effectiveness and efficiency of the proposed solution approach. We vary the number of intermediate PoI categories ($5$, $10$, $20$) to evaluate scalability, defaulting to $10$ for reporting. All experiments are averaged over three runs.

\paragraph{\textbf{Baseline Methodologies.}} \label{Sec:Baseline}
\begin{itemize}
\item \textbf{Random PoI and Random Medium (RPRM).}
In this approach, from source to destination via intermediate categories of PoI, the PoIs are selected randomly, and the transport medium between two PoIs is chosen randomly.
\item \textbf{Random PoI and Random Medium (RPRM).}
In this approach, from source to destination via intermediate categories of PoI, the PoIs are selected randomly, and the transport medium between two PoIs is chosen randomly.
\item \textbf{Random PoI and Cheapest Medium (RPCM).}
This approach selects PoIs randomly from source to destination via intermediate categories of PoI; however, consider the cheapest transport medium to visit one PoI to others.
\item \textbf{Nearest Neighbor PoI and Cheapest Medium (NNCM).}
This approach considered the nearest PoI in the first category of PoIs from the source using the cheapest transport medium, and from onwards, it selects one PoI from each category, considering the same till the destination.
\end{itemize}

\paragraph{\textbf{Goals of our Experiments.}} \label{Sec:Research_Questions}
The following research questions (RQ) are our focus in this study.
\begin{itemize}
\item \textbf{RQ1}: Varying agents, how do the travel cost and computational time change?
\item \textbf{RQ2}: Varying agents, how does the usage of transport medium change?
\item \textbf{RQ3}: Varying number of PoI in each category, how do the computational time and transport medium change?
\end{itemize}

\paragraph{\textbf{Experimental Results and Discussions.}}
In this section, we will address the research questions posed in Section \ref{Sec:Research_Questions} and discuss the experimental results.
\begin{figure*}[h!]
\centering
\begin{tabular}{cccc}
\includegraphics[width=0.25\linewidth]{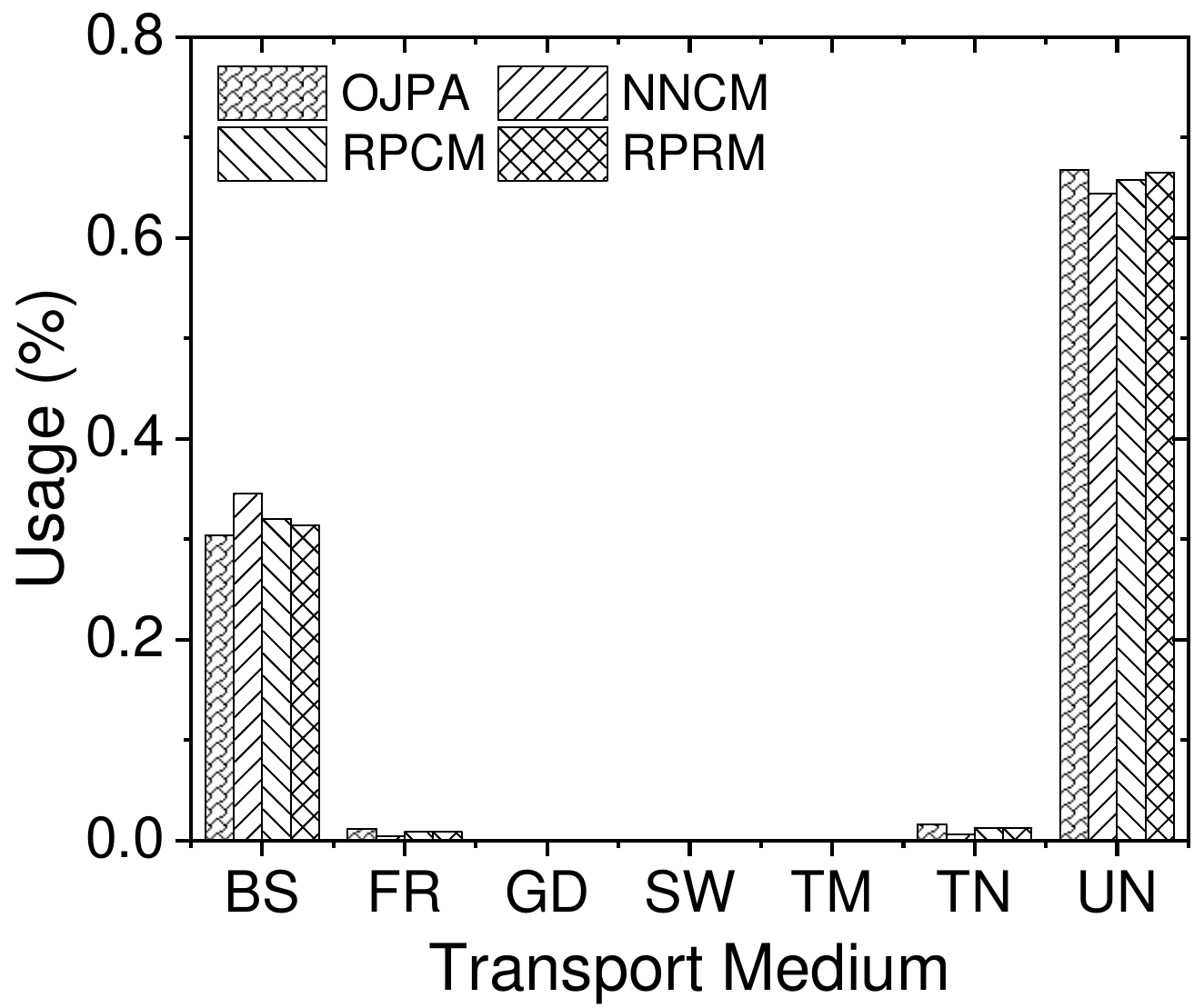}\hspace{-0.2em} & 
\includegraphics[width=0.25\linewidth]{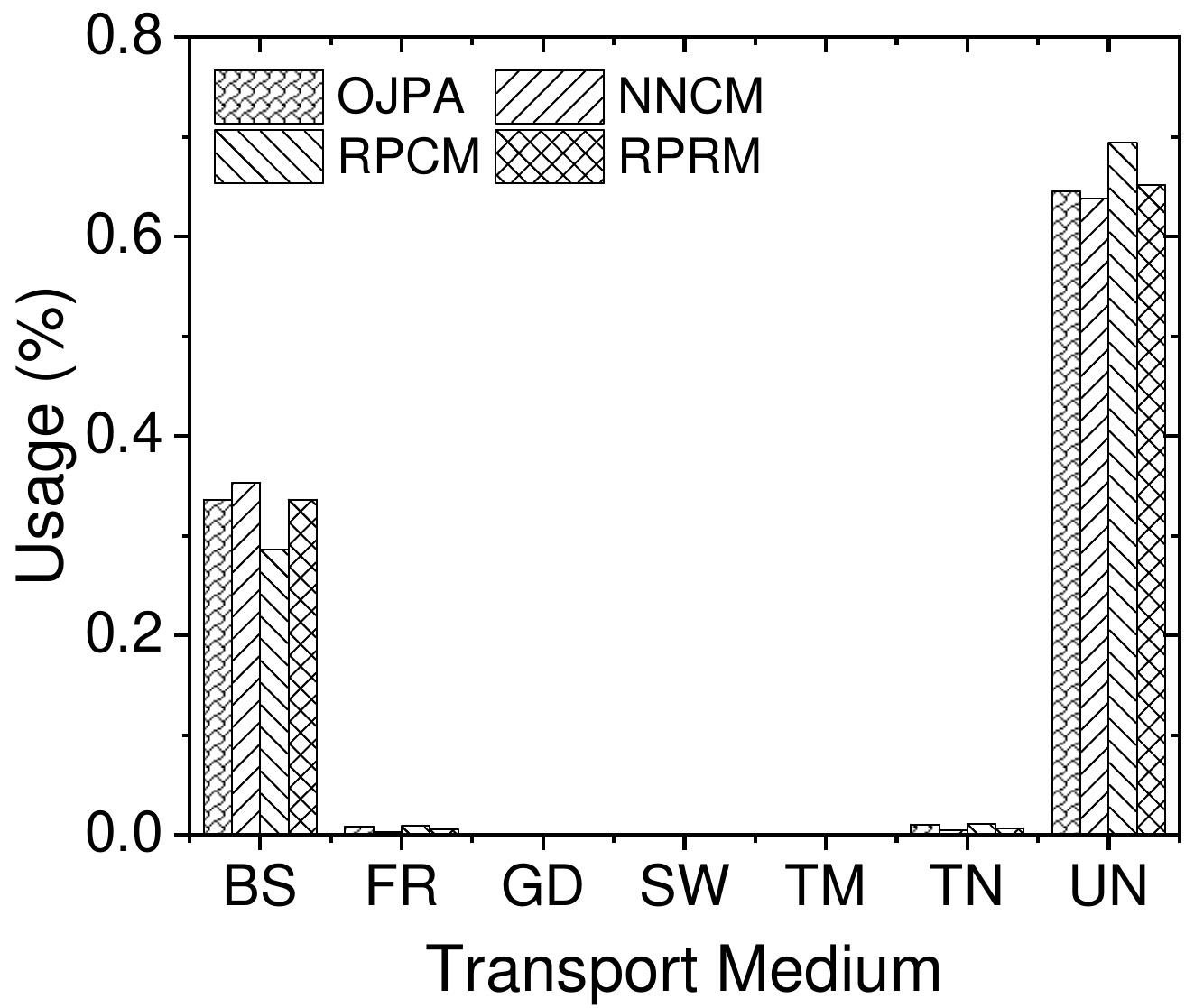}\hspace{-0.2em} & 
\includegraphics[width=0.25\linewidth]{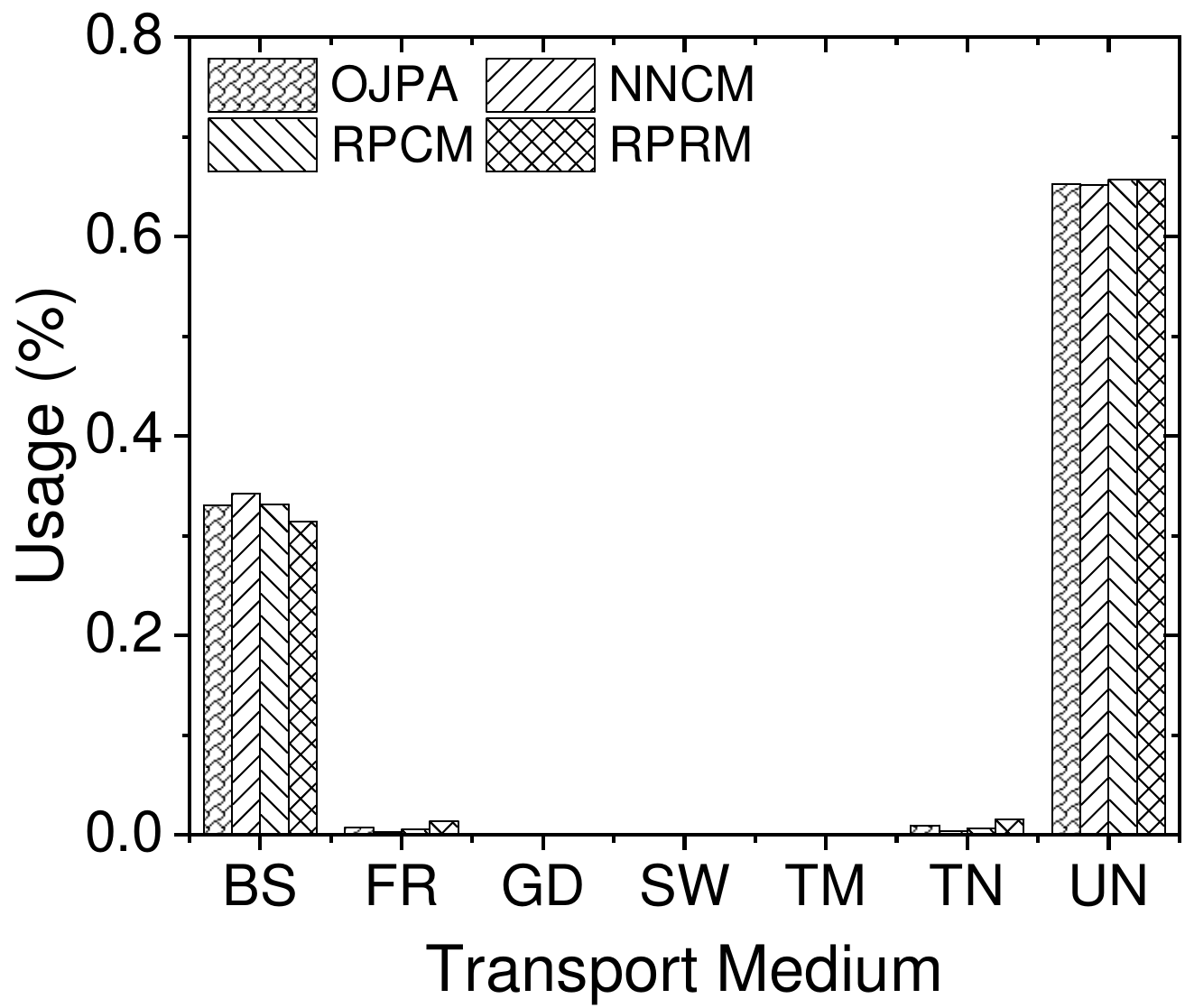}\hspace{-0.2em} & 
\includegraphics[width=0.25\linewidth]{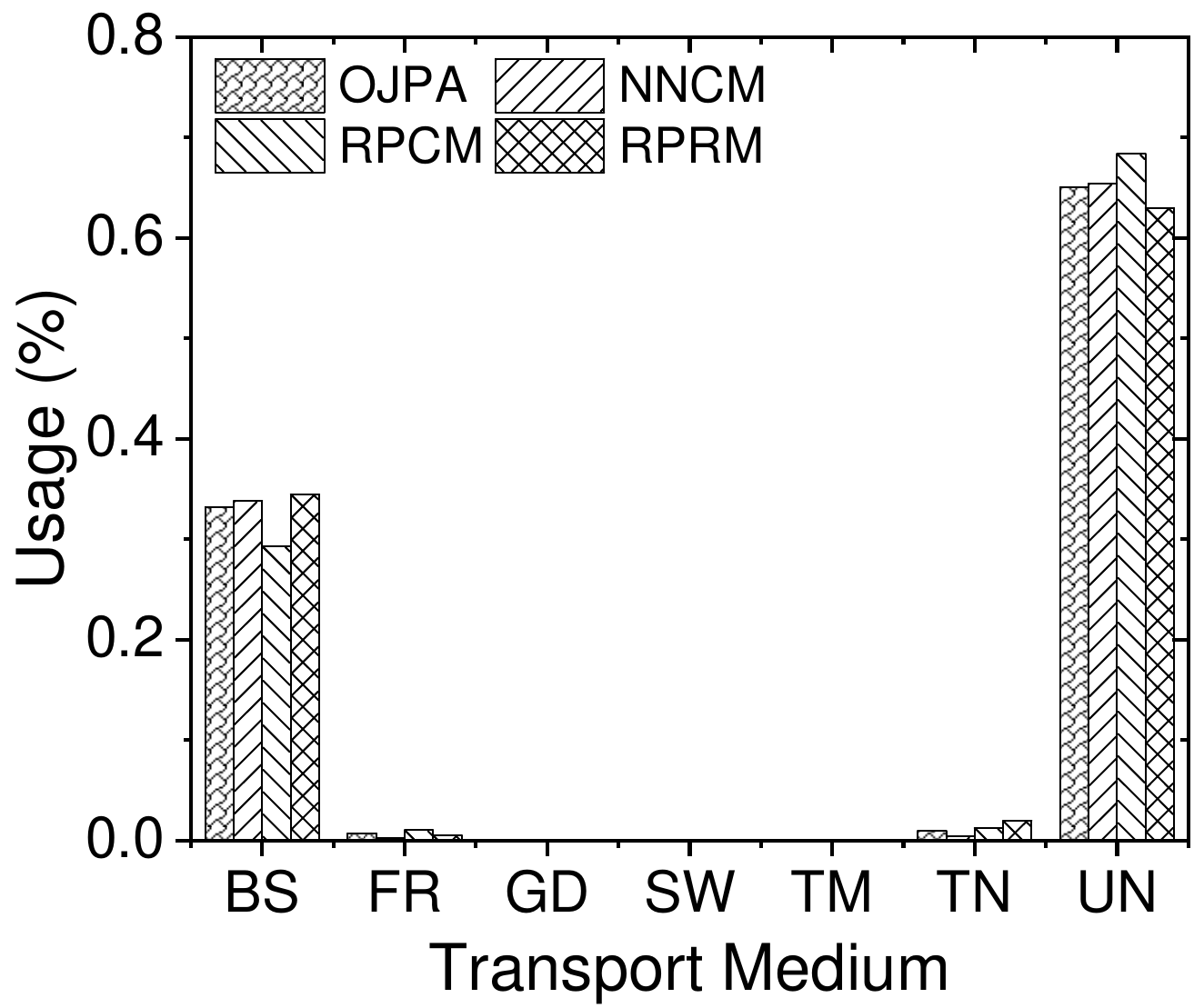} \\
\tiny{(a) $|\mathcal{U}| = 5$} &  \tiny{(b) $|\mathcal{U}| = 10$} & \tiny{(c) $|\mathcal{U}| = 20$} & \tiny{(d) $|\mathcal{U}| = 50$} \\
\includegraphics[width=0.25\linewidth]{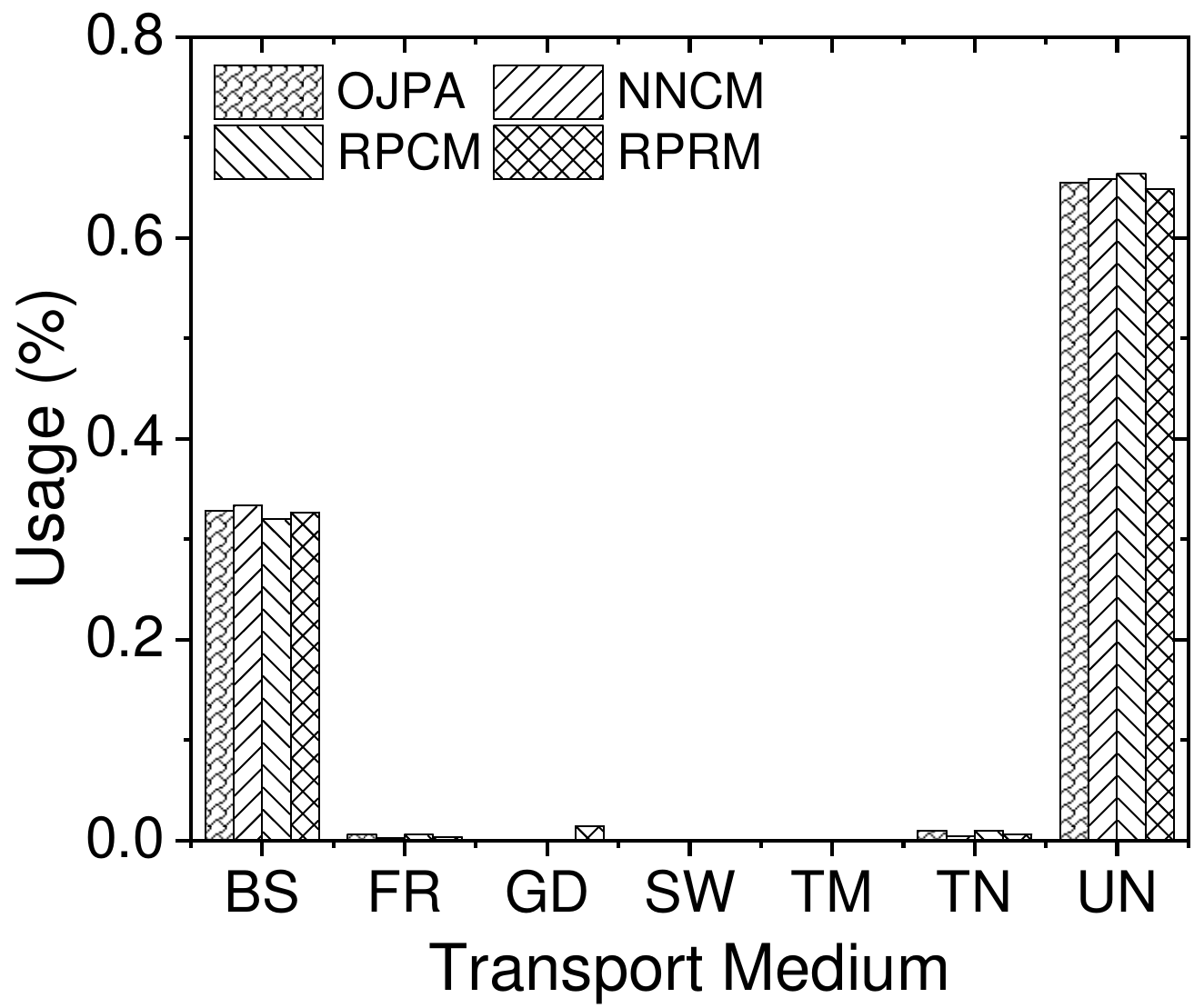}\hspace{-0.2em} &
\includegraphics[width=0.25\linewidth]{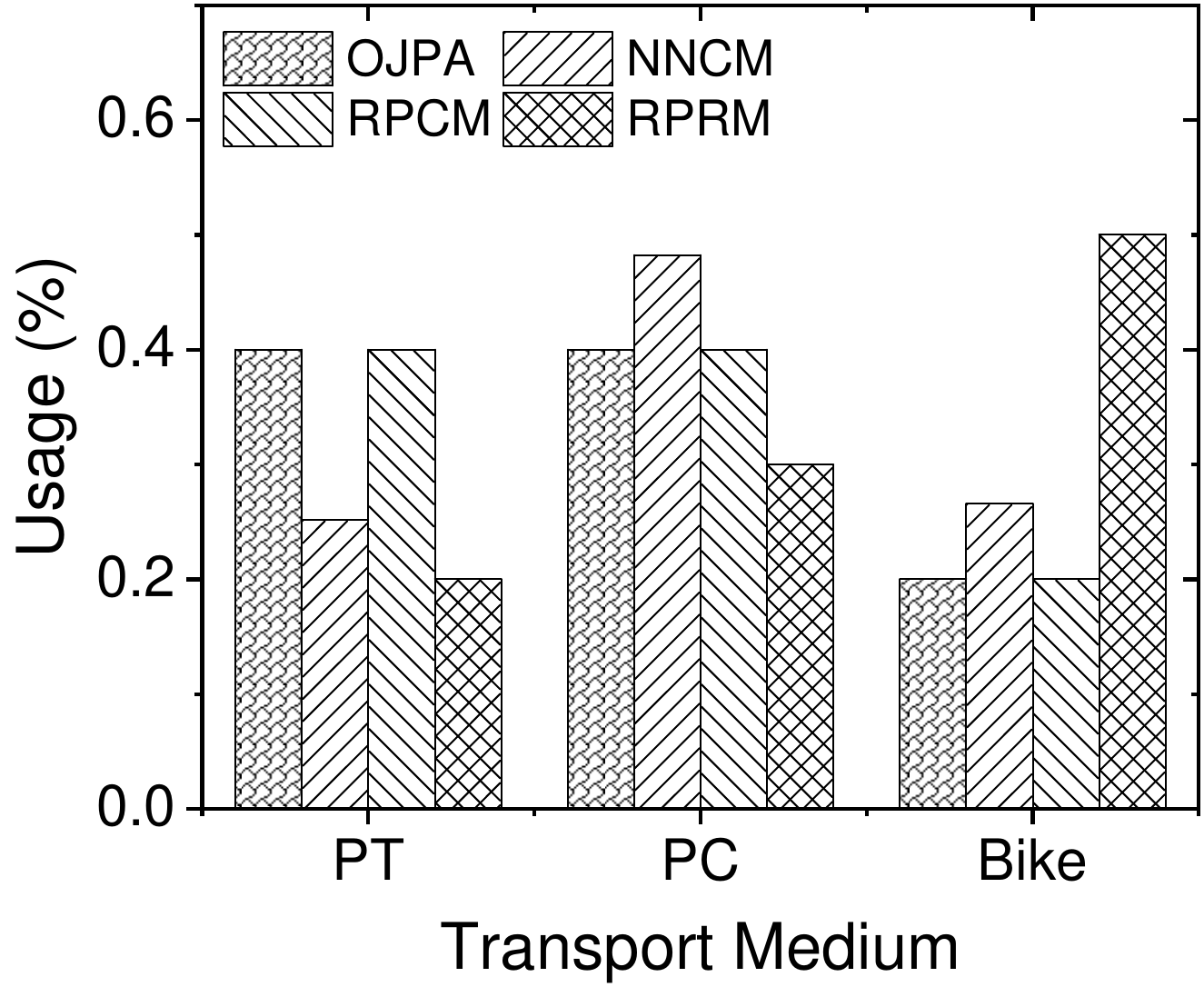}\hspace{-0.2em} &
\includegraphics[width=0.25\linewidth]{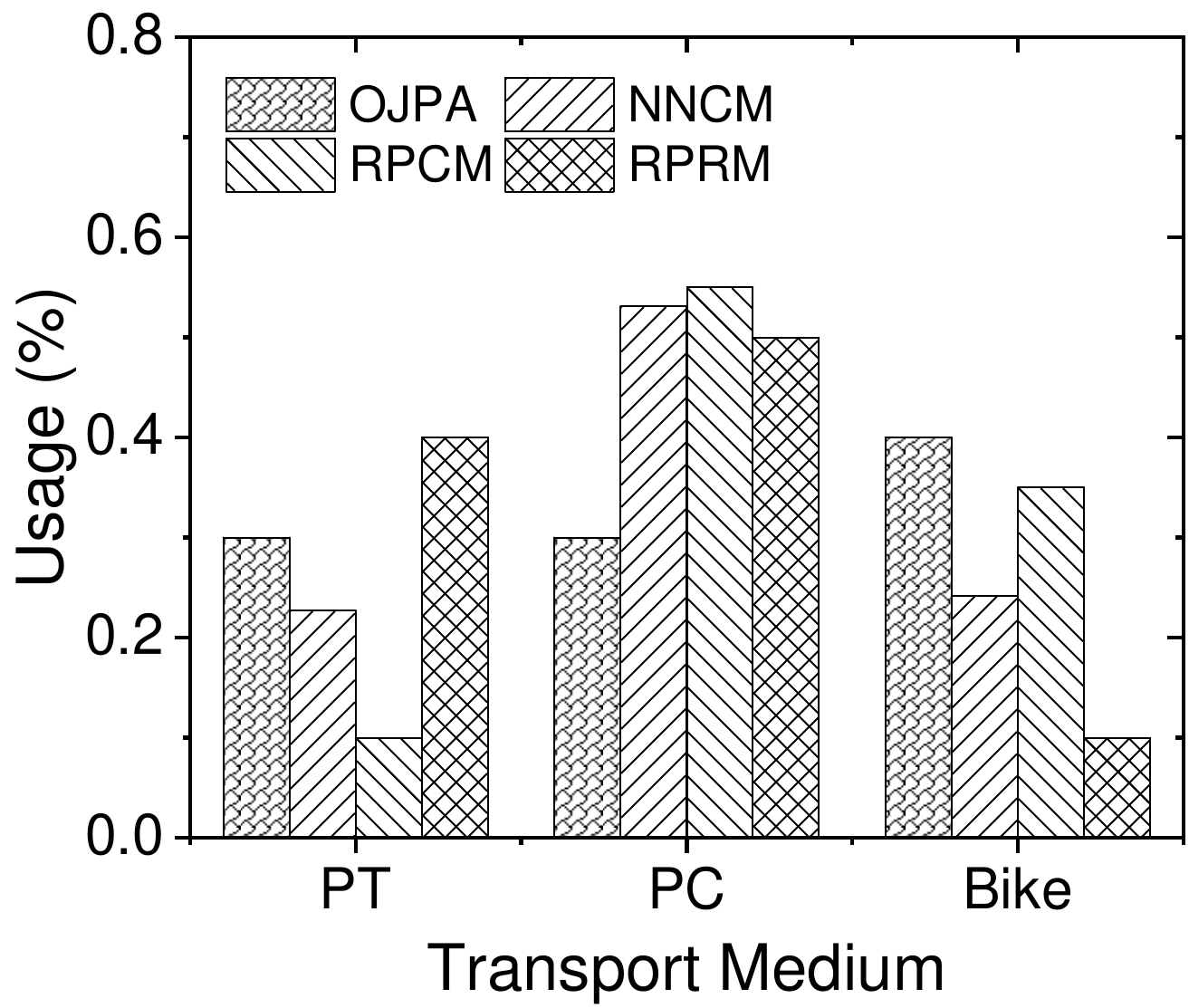}\hspace{-0.2em} &
\includegraphics[width=0.25\linewidth]{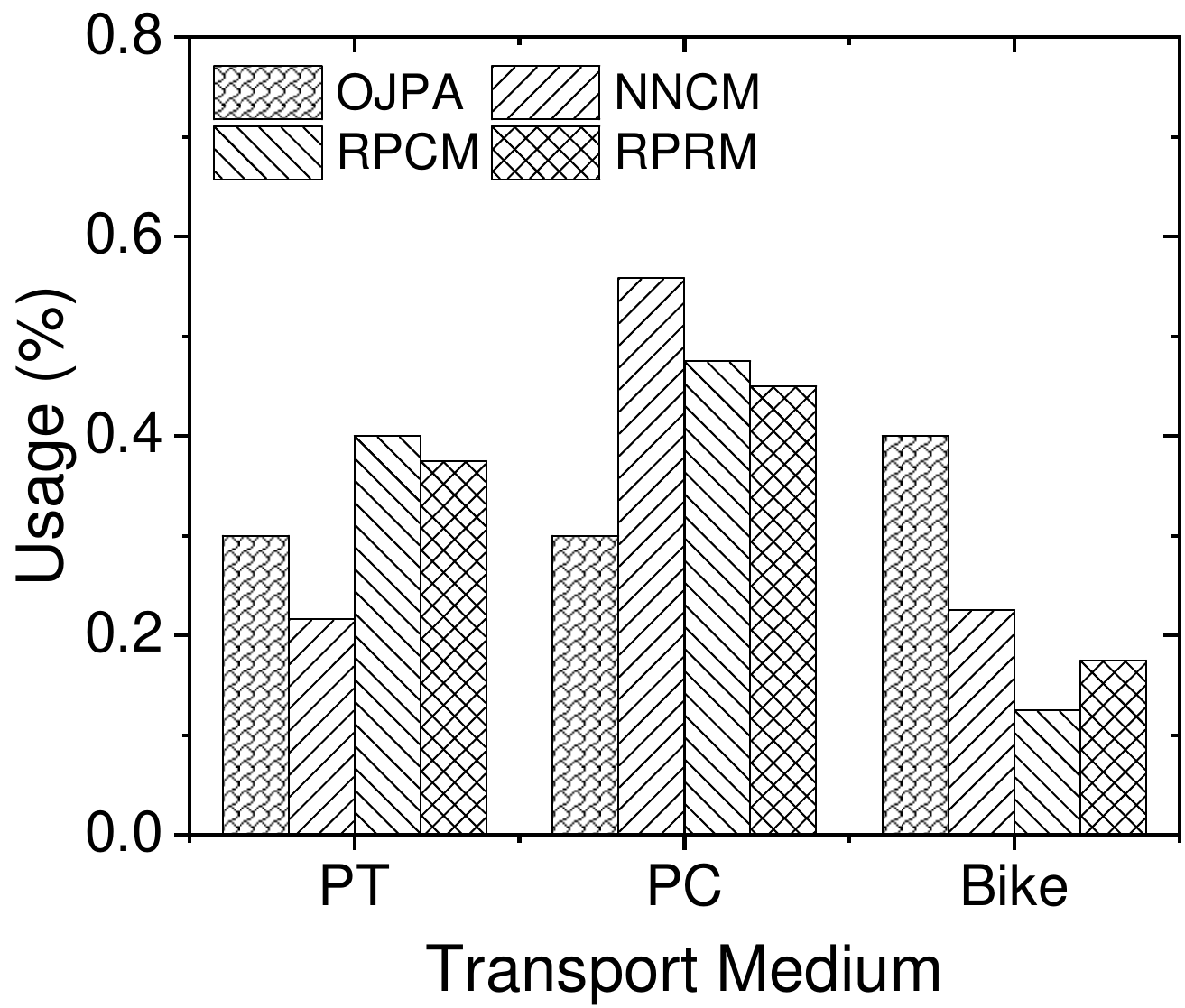}\\
\tiny{(e) $|\mathcal{U}| = 100$} & \tiny{(f) $|\mathcal{U}| = 5$} &  \tiny{(g) $|\mathcal{U}| = 10$} & \tiny{(h) $|\mathcal{U}| = 20$} \\
\includegraphics[width=0.25\linewidth]{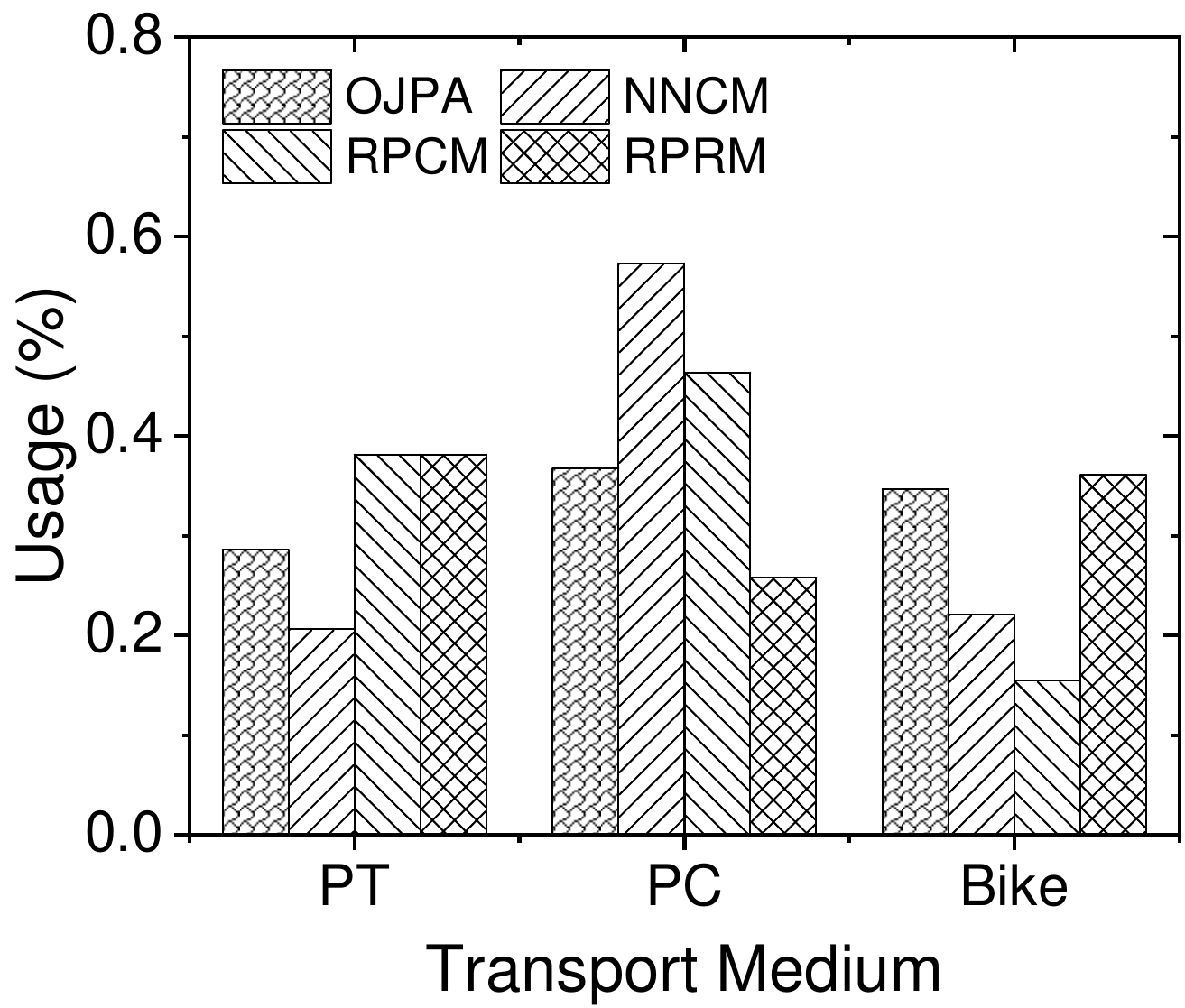}\hspace{-0.2em} &
\includegraphics[width=0.25\linewidth]{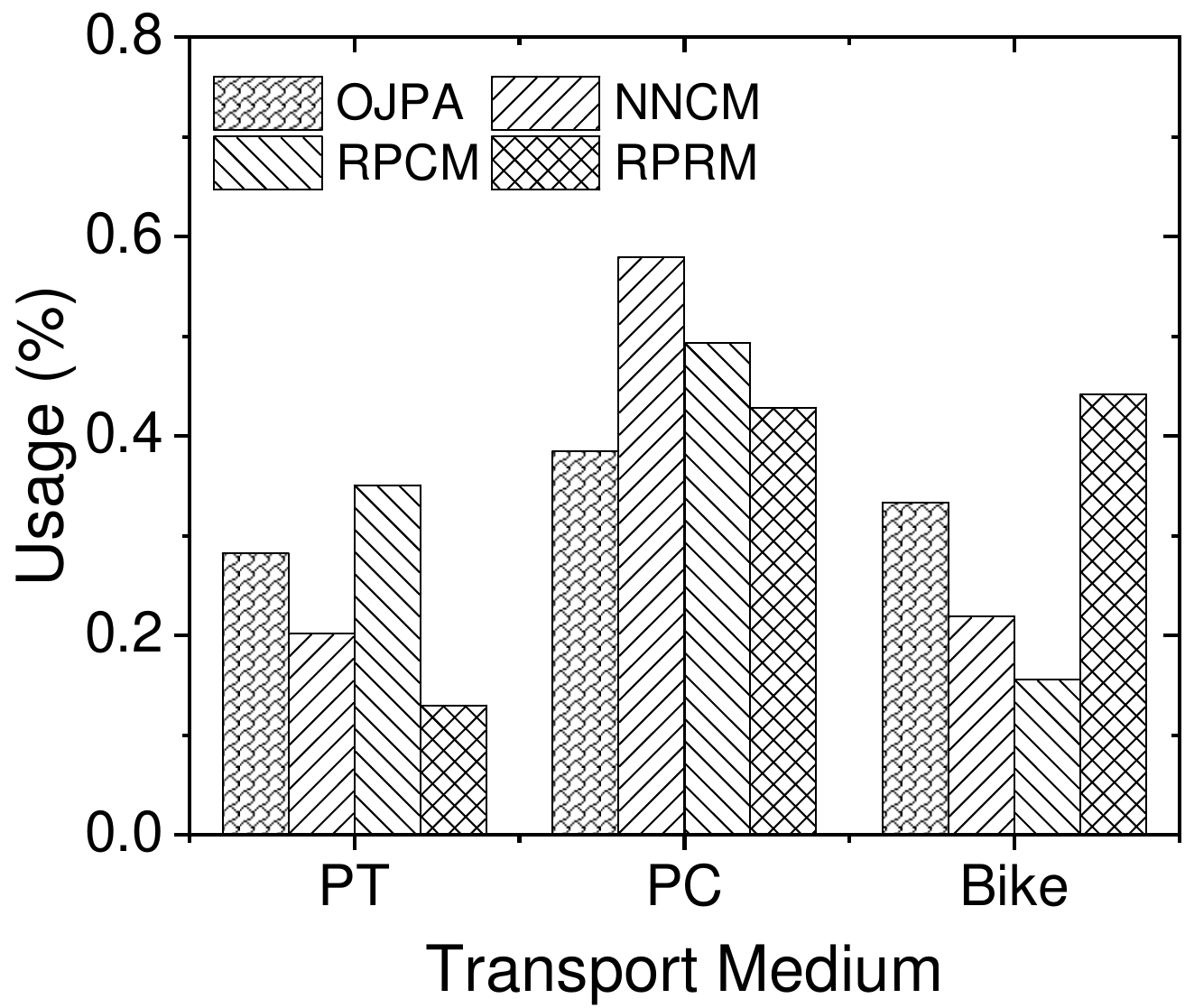}\hspace{-0.2em} &
\includegraphics[width=0.25\linewidth]{FA_Varying_PoI_Cost.pdf}\hspace{-0.2em} &
\includegraphics[width=0.26\linewidth]{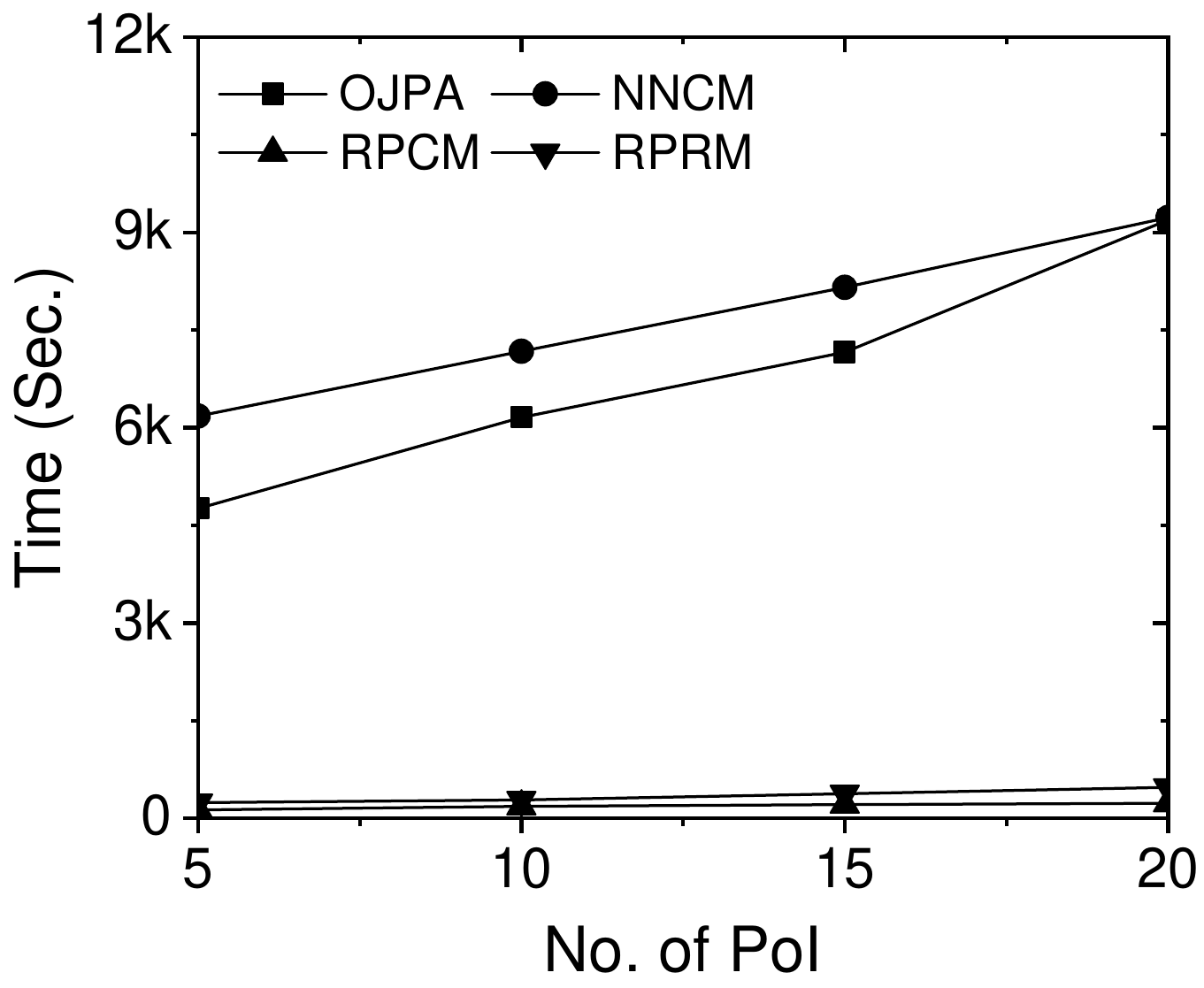}\\
\tiny{(i) $|\mathcal{U}| = 50$} & \tiny{(j) $|\mathcal{U}| = 100$}  & \tiny{(k) No. of PoI Vs. Cost}  & \tiny{$(\ell)$No. of PoI Vs. Time}
\end{tabular}
\caption{ Varying $|\mathcal{U}|$ Vs. Usage of transport medium in Switzerland $(a,b,c,d,e)$ and in Finland $(f,g,h,i,j)$, Varying No. of PoI Vs. Cost (k) and  Varying No. of PoI Vs. Time $(\ell)$  for Switzerland dataset}
\label{Fig:1Plot}
\end{figure*}

\begin{figure*}[h!]
\centering
\begin{tabular}{cccc}
 \includegraphics[width=0.25\linewidth]{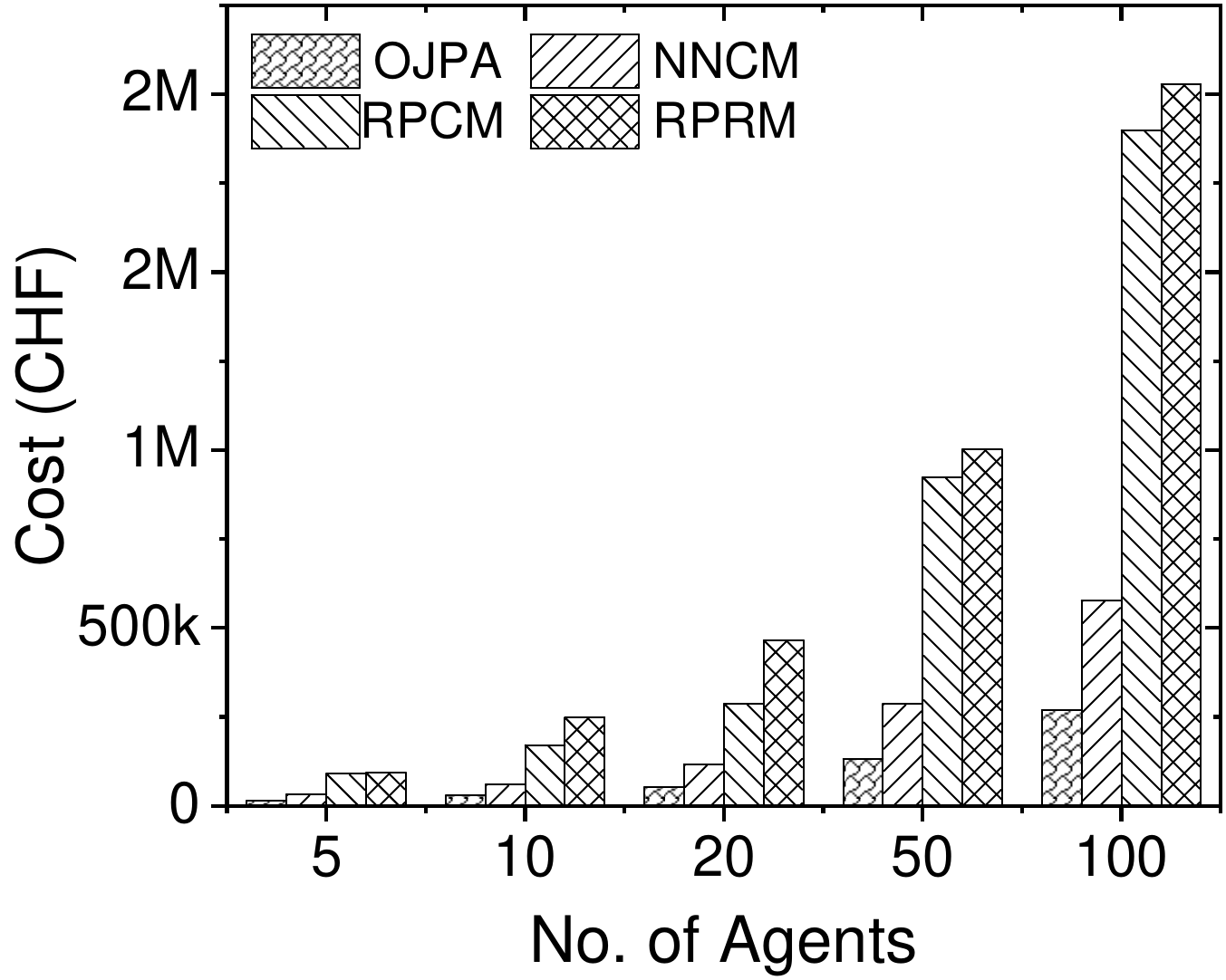}\hspace{-0.2em} &
\includegraphics[width=0.26\linewidth]{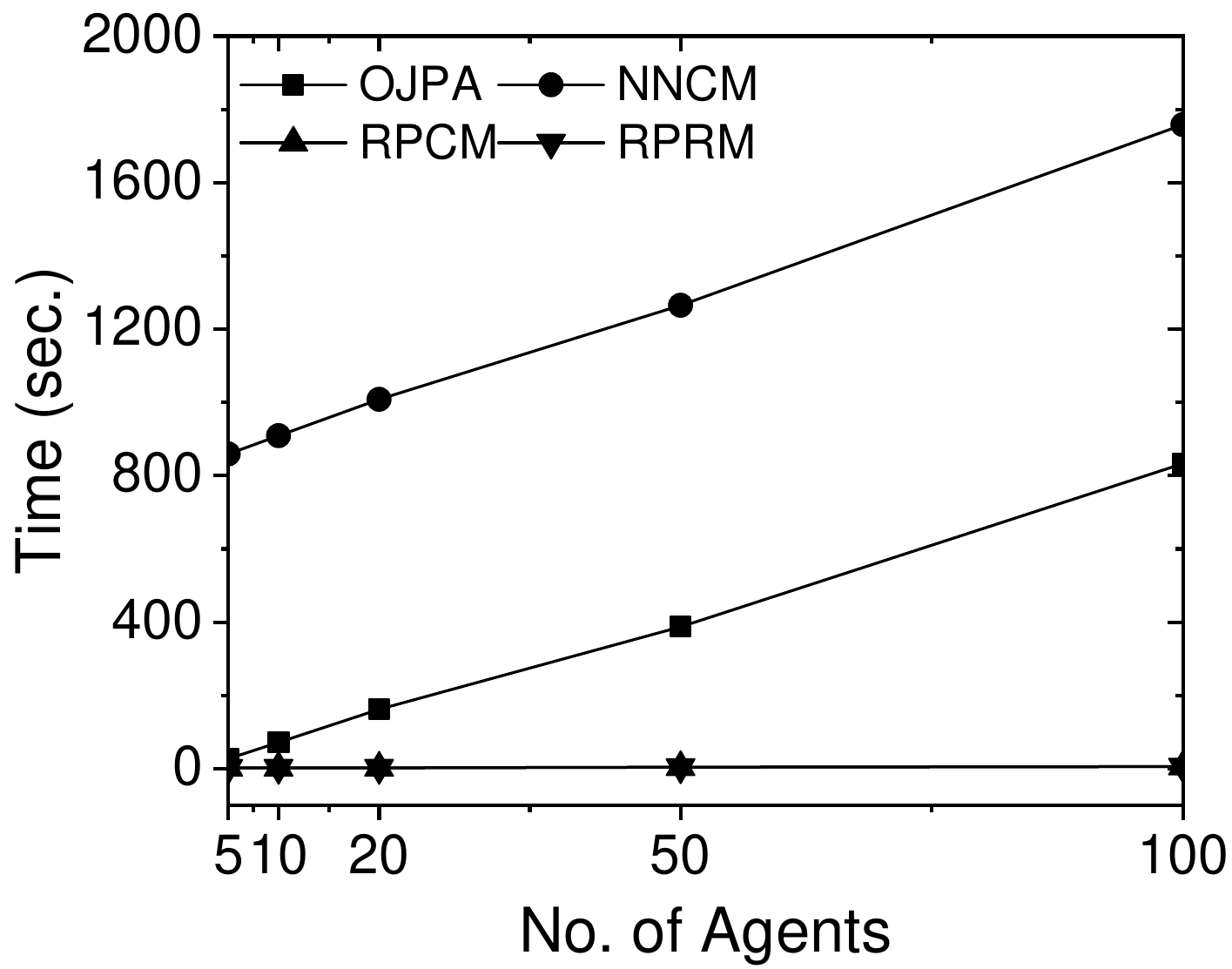}\hspace{-0.2em} &
\includegraphics[width=0.24\linewidth]{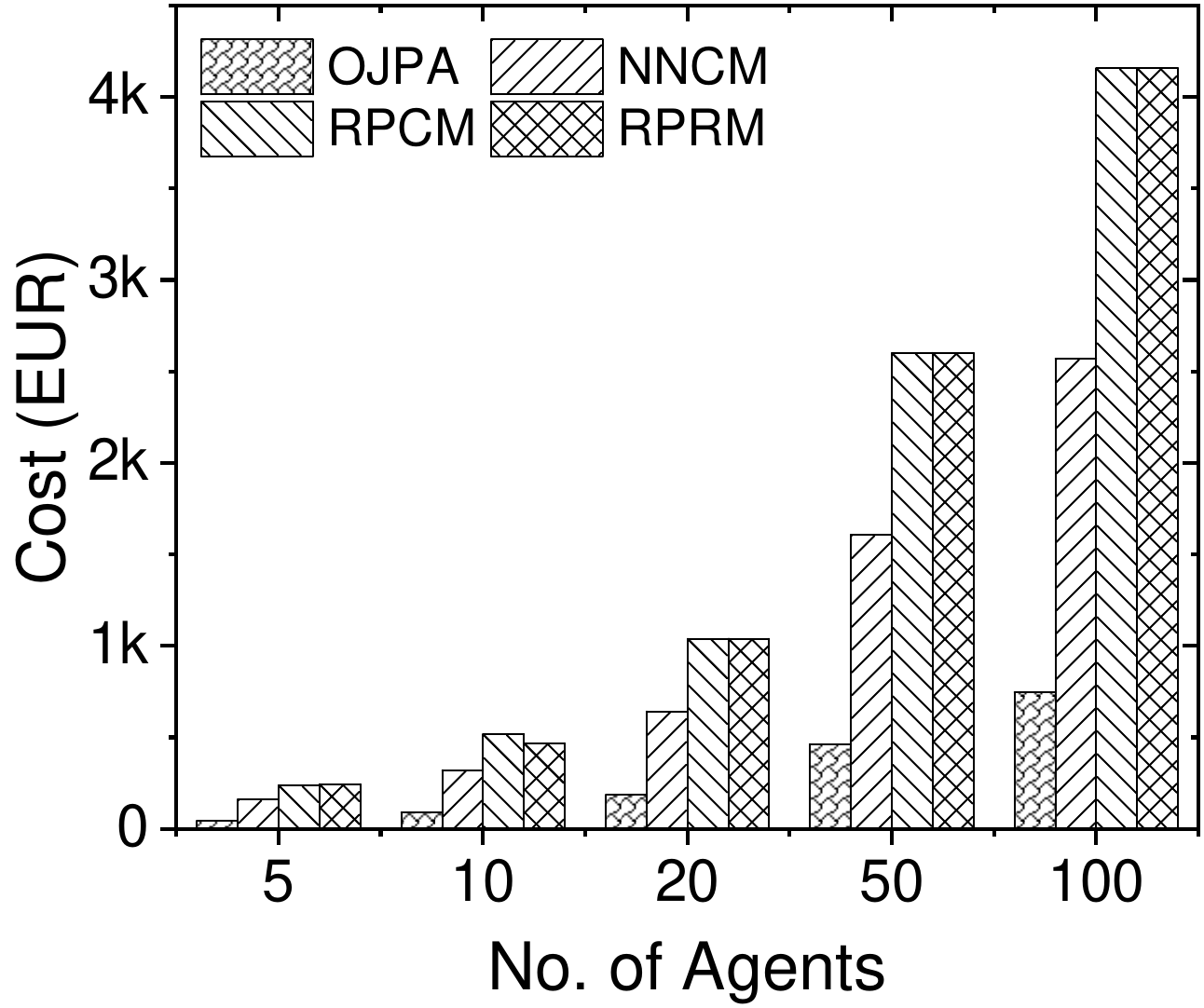}\hspace{-0.2em} &
\includegraphics[width=0.25\linewidth]{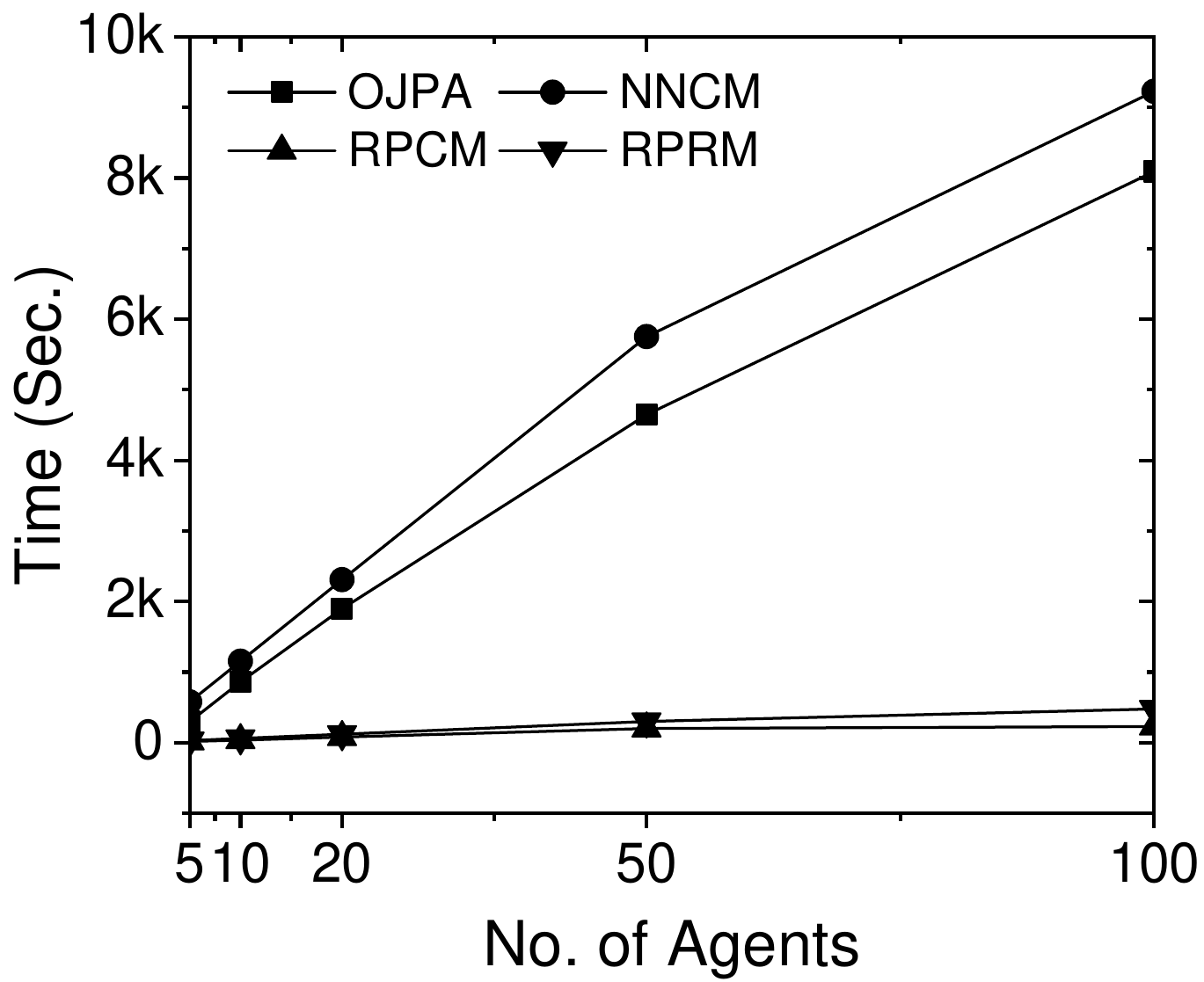} \\
\tiny{(a)  Varying $|\mathcal{U}|$ Vs. Cost}  & \tiny{(b) Varying $|\mathcal{U}|$ Vs. Time} & \tiny{(c)  Varying $|\mathcal{U}|$ Vs. Cost}  & \tiny{(d) Varying $|\mathcal{U}|$ Vs. Time} \\
\end{tabular}
\caption{ Varying $|\mathcal{U}|$ in Switzerland $(a,b)$ and in Finland $(c,d)$}
\label{Fig:2Plot}
\end{figure*}

\paragraph{\textsf{No. of Agents Vs. Travel Cost.}}
In this work, we vary the number of agents by $5$ to $100$ to evaluate the transport cost via the different journey mediums.In Figures \ref{Fig:2Plot}$(a)$ and \ref{Fig:2Plot}$(c)$, travel costs rise with the increasing number of agents, significantly driven by segments from the source to the first category and from the last category to the destinations. On the Finland dataset, the proposed `OJPA' consistently outperforms baseline methods (`NNCM', `RPCM', `RPRM') in terms of cost. Among the baselines, `NNCM' achieves the lowest cost by always selecting the nearest PoI and cheapest travel medium, whereas `RPRM' incurs the highest cost due to its random selection strategy. Similar observations are noted for the Switzerland dataset, as shown in Figure \ref{Fig:2Plot}$(c)$. For instance, on the Finland dataset, with $5$ agents, travel costs for `OJPA', `NNCM', `RPCM', and `RPRM' are $42.75$ € , $160.765$ €, $237.5$ €, and $244.6$ €, respectively. With $100$ agents, the costs rise to $745.75$ €, $2569.23$ €, $4156.25$ €, and $4156.25$ €, respectively. For the Switzerland dataset, as the number of agents increases from $5$ to $100$, the travel costs for `OJPA', `NNCM', `RPCM', and `RPRM' rise sharply from $30,701.25$ CHF, $68,861.25$ CHF, $308,353.75$ CHF, and $204,588.75$ CHF to $15,005,950$ CHF, $75,782,150$ CHF, $103,046,875$ CHF, and $70,440,000$ CHF, respectively. The substantially higher costs compared to the Finland dataset result from the lack of direct paths between PoIs in Switzerland's road network, whereas such direct connections are more commonly available in Finland.

\paragraph{\textsf{No. of Agents Vs. Time.}}
To determine the computational time required, we vary the number of agents and observe the run time is proportional to a varying number of agents. We have compared our proposed approach with baseline methods like `NNCM,' `RPCM,' and `RPRM'; among them, `NNCM' takes more time than the other baseline methods. This happens because `NNCM', finds the nearest PoI from one category to the PoIs from another. It is also considered the cheapest medium of the journey, which leads to huge computational time. However, the `RPCM' and `RPRM' both methods use randomization to select PoIs from one category to another, leading to very less runtime, and it is negligible compared to the other methods. These observations are consistent with both Switzerland as well as Finland datasets, as shown in Figure \ref{Fig:2Plot}$(b)$ and \ref{Fig:2Plot}$(d)$, respectively. For example, in the Switzerland dataset, when we vary agents from $5$ to $100$, the runtime for `OJPA', `NNCM', `RPCM', and `RPRM' varies from $24,25,2,2$ to $546,605,58,53$ in seconds, respectively. 

\paragraph{\textsf{No. of Agents Vs. Medium Usage.}}
Figure \ref{Fig:1Plot}$(a)$ to \ref{Fig:1Plot}$(j)$ illustrates how the usage of different transport mediums varies with the number of agents in the Switzerland dataset, considering mediums such as bus (BS), ferry (FR), gondola (GD), subway (SW), tram (TM), and train (TN). As the number of agents increases, the usage of transport mediums also rises, but only bus, ferry, gondola, and tram are utilized across both proposed and baseline methods. Among these, the bus is consistently the most frequently used medium, whereas the gondola is the least utilized. Figures \ref{Fig:1Plot}$(a)$–$(e)$ illustrate the average usage of travel mediums, clearly indicating that the Unknown (UN) medium has the highest usage. This occurs because the original road network lacks sufficient connectivity, requiring additional edges to form a connected graph. Agents frequently select these added edges for their optimal paths. The proposed `OJPA' and baseline `NNCM' methods show similar usage patterns for different travel mediums, whereas the random selection in `RPCM' and `RPRM' results in significantly varied medium usage. The Finland dataset consists of three journey modes: public transport (PT), private car (PC), and bike, and they are used very frequently. Public transport and private cars are the most used medium, as reported in Figure \ref{Fig:1Plot}$(f,g,h,i,j)$. We have observed that with the increase in the number of agents, the usage of private cars as a medium increases and public transport usage decreases. These observations are consistent with the Switzerland dataset.

\paragraph{\textsf{Additional Discussion.}}
Additionally, we have experimented with varying numbers of PoIs in each category to check the run time and usage of different mediums of the journey. We fixed the number of agents as $100$ and varied the number of PoIs in each category by $5, 10, 15,$, and $20$, and the experimental results are reported in Figure \ref{Fig:1Plot}$(k, \ell)$. We have observed that the computational time increases rapidly with the increased number of agents. Now, in the case of medium usage, minor changes happen in the case of the `OJPA' and `NNCM' approaches. However, major changes in usage occur in the `RPCM', and `RPRM'. We have also varied the number of intermediate categories by $5,10,20$ and observed that with an increasing number of categories, the computational cost increases rapidly. One point needs to be noted in the Switzerland and Finland datasets: we have considered an equal number of PoIs in each category, including the source and destination categories for all our experiments as default settings.

\section{Concluding Remarks} \label{Sec:CFD}
In this paper, we studied the GTP query problem considering multiple commuting transport mediums, a direction previously unexplored to the best of our knowledge. We demonstrated that increasing the number of PoI categories makes the problem computationally intractable. We proposed a dynamic programming-based solution and analyzed its time and space complexity to address this. Extensive experiments on publicly available benchmark datasets revealed that our method consistently outperforms several baseline approaches. An important future research direction includes addressing fairness concerns, as minimizing the aggregated group cost may disproportionately affect individual agents.

\bibliographystyle{splncs04}
\bibliography{Paper}

\begin{thebibliography}{10}
\providecommand{\url}[1]{\texttt{#1}}
\providecommand{\urlprefix}{URL }
\providecommand{\doi}[1]{https://doi.org/#1}

\bibitem{ahmadi2015mixed}
Ahmadi, E., Nascimento, M.A.: A mixed breadth-depth first search strategy for
  sequenced group trip planning queries. In: 2015 16th IEEE International
  Conference on Mobile Data Management. vol.~1, pp. 24--33. IEEE (2015)

\bibitem{barua2017weighted}
Barua, S., Jahan, R., Ahmed, T.: Weighted optimal sequenced group trip planning
  queries. In: 2017 18th IEEE International Conference on Mobile Data
  Management (MDM). pp. 222--227. IEEE (2017)

\bibitem{ceder2016public}
Ceder, A.: Public transit planning and operation: Modeling, practice and
  behavior. CRC press (2016)

\bibitem{delling2008timetable}
Delling, D., Giannakopoulou, K., Wagner, D., Zaroliagis, C.: Timetable
  information updating in case of delays: Modeling issues. Arrival Technical
  Report  (2008)

\bibitem{delling2015round}
Delling, D., Pajor, T., Werneck, R.F.: Round-based public transit routing.
  Transportation Science  \textbf{49}(3),  591--604 (2015)

\bibitem{delling2009engineering}
Delling, D., Sanders, P., Schultes, D., Wagner, D.: Engineering route planning
  algorithms. In: Algorithmics of large and complex networks: design, analysis,
  and simulation, pp. 117--139. Springer (2009)

\bibitem{dijkstra2022note}
Dijkstra, E.W.: A note on two problems in connexion with graphs. In: Edsger
  Wybe Dijkstra: his life, work, and legacy, pp. 287--290 (2022)

\bibitem{hashem2013group}
Hashem, T., Hashem, T., Ali, M.E., Kulik, L.: Group trip planning queries in
  spatial databases. In: Advances in Spatial and Temporal Databases: 13th
  International Symposium, SSTD 2013, Munich, Germany, August 21-23, 2013.
  Proceedings 13. pp. 259--276. Springer (2013)

\bibitem{lee2020collective}
Lee, J., Park, S.: The collective trip planning query processing using g-tree
  index structure. In: 2020 IEEE International Conference on Big Data and Smart
  Computing (BigComp). pp. 173--180. IEEE (2020)

\bibitem{li2005trip}
Li, F., Cheng, D., Hadjieleftheriou, M., Kollios, G., Teng, S.H.: On trip
  planning queries in spatial databases. In: International symposium on spatial
  and temporal databases. pp. 273--290. Springer (2005)

\bibitem{potthoff2022efficient}
Potthoff, M., Sauer, J.: Efficient algorithms for fully multimodal journey
  planning. In: 22nd Symposium on Algorithmic Approaches for Transportation
  Modelling, Optimization, and Systems (ATMOS 2022). Schloss-Dagstuhl-Leibniz
  Zentrum f{\"u}r Informatik (2022)

\bibitem{sauer2020efficient}
Sauer, J., Wagner, D., Z{\"u}ndorf, T.: An efficient solution for one-to-many
  multi-modal journey planning. In: 20th Symposium on Algorithmic Approaches
  for Transportation Modelling, Optimization, and Systems (ATMOS 2020).
  Schloss-Dagstuhl-Leibniz Zentrum f{\"u}r Informatik (2020)

\bibitem{shekhar2007spatial}
Shekhar, S.: Spatial databases. Pearson Education India (2007)

\bibitem{solanki2023fairness}
Solanki, N., Jain, S., Banerjee, S., Kumar~S, Y.P.: Fairness driven efficient
  algorithms for sequenced group trip planning query problem. In: Proceedings
  of the 2023 International Conference on Autonomous Agents and Multiagent
  Systems. pp. 86--94 (2023)

\bibitem{tabassum2017dynamic}
Tabassum, A., Barua, S., Hashem, T., Chowdhury, T.: Dynamic group trip planning
  queries in spatial databases. In: Proceedings of the 29th international
  conference on scientific and statistical database management. pp.~1--6 (2017)

\bibitem{tenkanen2020longitudinal}
Tenkanen, H., Toivonen, T.: Longitudinal spatial dataset on travel times and
  distances by different travel modes in helsinki region. Scientific data
  \textbf{7}(1), ~77 (2020)

\end{thebibliography}
\end{document}